\documentclass[conference,a4paper]{IEEEtran}
\usepackage{longtable}
\usepackage{graphicx} 
\usepackage{multicol}
\usepackage{xcolor,colortbl}
\usepackage{mleftright}
\usepackage{cite}
\usepackage{enumerate}
\usepackage{amsmath}
\usepackage{amssymb}
\usepackage{amsfonts}
\usepackage{amsthm}
\usepackage{mathrsfs}
\usepackage{xspace}
\usepackage{bm}
\usepackage{fancyref}
\usepackage{textcomp}
\usepackage[Symbol]{upgreek}

\newcommand{\safemath}[2]{\newcommand{#1}{\ensuremath{#2}\xspace}}

\newcommand{\ssa}{\mathsf{a}}
\newcommand{\ssb}{\mathsf{b}}
\newcommand{\ssc}{\mathsf{c}}
\newcommand{\ssd}{\mathsf{d}}
\newcommand{\sse}{\mathsf{e}}
\newcommand{\ssf}{\mathsf{f}}
\newcommand{\ssg}{\mathsf{g}}
\newcommand{\ssh}{\mathsf{h}}
\newcommand{\ssi}{\mathsf{i}}
\newcommand{\ssj}{\mathsf{j}}
\newcommand{\ssk}{\mathsf{k}}
\newcommand{\ssl}{\mathsf{l}}
\newcommand{\ssm}{\mathsf{m}}
\newcommand{\ssn}{\mathsf{n}}
\newcommand{\sso}{\mathsf{o}}
\newcommand{\ssp}{\mathsf{p}}
\newcommand{\ssq}{\mathsf{q}}
\newcommand{\ssr}{\mathsf{r}}
\newcommand{\sss}{\mathsf{s}}
\newcommand{\sst}{\mathsf{t}}
\newcommand{\ssu}{\mathsf{u}}
\newcommand{\ssv}{\mathsf{v}}
\newcommand{\ssw}{\mathsf{w}}
\newcommand{\ssx}{\mathsf{x}}
\newcommand{\ssy}{\mathsf{y}}
\newcommand{\ssz}{\mathsf{z}}

\safemath{\bmsa}{\bm{\ssa}}
\safemath{\bmsb}{\bm{\ssb}}
\safemath{\bmsc}{\bm{\ssc}}
\safemath{\bmsd}{\bm{\ssd}}
\safemath{\bmse}{\bm{\sse}}
\safemath{\bmsf}{\bm{\ssf}}
\safemath{\bmsg}{\bm{\ssg}}
\safemath{\bmsh}{\bm{\ssh}}
\safemath{\bmsi}{\bm{\ssi}}
\safemath{\bmsj}{\bm{\ssj}}
\safemath{\bmsk}{\bm{\ssk}}
\safemath{\bmsl}{\bm{\ssl}}
\safemath{\bmsm}{\bm{\ssm}}
\safemath{\bmsn}{\bm{\ssn}}
\safemath{\bmso}{\bm{\sso}}
\safemath{\bmsp}{\bm{\ssp}}
\safemath{\bmsq}{\bm{\ssq}}
\safemath{\bmsr}{\bm{\ssr}}
\safemath{\bmss}{\bm{\sss}}
\safemath{\bmst}{\bm{\sst}}
\safemath{\bmsu}{\bm{\ssu}}
\safemath{\bmsv}{\bm{\ssv}}
\safemath{\bmsw}{\bm{\ssw}}
\safemath{\bmsx}{\bm{\ssx}}
\safemath{\bmsy}{\bm{\ssy}}
\safemath{\bmsz}{\bm{\ssz}}

\bmdefine{\bmualphad}{\upalpha}
\bmdefine{\bmubetad}{\upbeta}
\bmdefine{\bmuchid}{\upchi}
\bmdefine{\bmudeltad}{\updelta}
\bmdefine{\bmuepsilond}{\upepsilon}
\bmdefine{\bmuvarepsilond}{\upvarepsilon}
\bmdefine{\bmuetad}{\upeta}
\bmdefine{\bmugammad}{\upgamma}
\bmdefine{\bmuiotad}{\upiota}
\bmdefine{\bmukappad}{\upkappa}
\bmdefine{\bmulambdad}{\uplambda}
\bmdefine{\bmumud}{\upmu}
\bmdefine{\bmunud}{\upnu}
\bmdefine{\bmuomegad}{\upomega}
\bmdefine{\bmuphid}{\upphi}
\bmdefine{\bmuvarphid}{\upvarphi}
\bmdefine{\bmupid}{\uppi}
\bmdefine{\bmuvarpid}{\upvarpi}
\bmdefine{\bmupsid}{\uppsi}
\bmdefine{\bmurhod}{\uprho}
\bmdefine{\bmuvarrhod}{\upvarrho}
\bmdefine{\bmusigmad}{\upsigma}
\bmdefine{\bmuvarsigmad}{\upvarsigma}
\bmdefine{\bmutaud}{\uptau}
\bmdefine{\bmuthetad}{\uptheta}
\bmdefine{\bmuvarthetad}{\upvartheta}
\bmdefine{\bmuupsilond}{\upupsilon}
\bmdefine{\bmuxid}{\upxi}
\bmdefine{\bmuzetad}{\upzeta}

\safemath{\bmua}{\mathbf{a}}
\safemath{\bmub}{\mathbf{b}}
\safemath{\bmuc}{\mathbf{c}}
\safemath{\bmud}{\mathbf{d}}
\safemath{\bmue}{\mathbf{e}}
\safemath{\bmuf}{\mathbf{f}}
\safemath{\bmug}{\mathbf{g}}
\safemath{\bmuh}{\mathbf{h}}
\safemath{\bmui}{\mathbf{i}}
\safemath{\bmuj}{\mathbf{j}}
\safemath{\bmuk}{\mathbf{k}}
\safemath{\bmul}{\mathbf{l}}
\safemath{\bmum}{\mathbf{m}}
\safemath{\bmun}{\mathbf{n}}
\safemath{\bmuo}{\mathbf{o}}
\safemath{\bmup}{\mathbf{p}}
\safemath{\bmuq}{\mathbf{q}}
\safemath{\bmur}{\mathbf{r}}
\safemath{\bmus}{\mathbf{s}}
\safemath{\bmut}{\mathbf{t}}
\safemath{\bmuu}{\mathbf{u}}
\safemath{\bmuv}{\mathbf{v}}
\safemath{\bmuw}{\mathbf{w}}
\safemath{\bmux}{\mathbf{x}}
\safemath{\bmuy}{\mathbf{y}}
\safemath{\bmuz}{\mathbf{z}}

\safemath{\bmualpha}{\bmualphad}
\safemath{\bmubeta}{\bmubetad}
\safemath{\bmuchi}{\bumchid}
\safemath{\bmudelta}{\bmudeltad}
\safemath{\bmuepsilon}{\bmuepsilond}
\safemath{\bmuvarepsilon}{\bmuvarepsilond}
\safemath{\bmueta}{\bmuetad}
\safemath{\bmugamma}{\bmugammad}
\safemath{\bmuiota}{\bmuiotad}
\safemath{\bmukappa}{\bmukappad}
\safemath{\bmulambda}{\bmulambdad}
\safemath{\bmumu}{\bmumud}
\safemath{\bmunu}{\bmunud}
\safemath{\bmuomega}{\bmuomegad}
\safemath{\bmuphi}{\bmuphid}
\safemath{\bmuvarphi}{\bmuvarphid}
\safemath{\bmupi}{\bmupid}
\safemath{\bmuvarpi}{\bmuvarpid}
\safemath{\bmupsi}{\bmupsid}
\safemath{\bmurho}{\bmurhod}
\safemath{\bmuvarrho}{\bmuvarrhod}
\safemath{\bmusigma}{\bmusigmad}
\safemath{\bmuvarsigma}{\bmuvarsigmad}
\safemath{\bmutau}{\bmutaud}
\safemath{\bmutheta}{\bmuthetad}
\safemath{\bmuvartheta}{\bmuvarthetad}
\safemath{\bmuupsilon}{\bmuupsilond}
\safemath{\bmuxi}{\bmuxid}
\safemath{\bmuzeta}{\bmuzetad}

\bmdefine{\bmiad}{a}
\bmdefine{\bmibd}{b}
\bmdefine{\bmicd}{c}
\bmdefine{\bmidd}{d}
\bmdefine{\bmied}{e}
\bmdefine{\bmifd}{f}
\bmdefine{\bmigd}{g}
\bmdefine{\bmihd}{h}
\bmdefine{\bmiid}{i}
\bmdefine{\bmijd}{j}
\bmdefine{\bmikd}{k}
\bmdefine{\bmild}{l}
\bmdefine{\bmimd}{m}
\bmdefine{\bmind}{n}
\bmdefine{\bmiod}{o}
\bmdefine{\bmipd}{p}
\bmdefine{\bmiqd}{q}
\bmdefine{\bmird}{r}
\bmdefine{\bmisd}{s}
\bmdefine{\bmitd}{t}
\bmdefine{\bmiud}{u}
\bmdefine{\bmivd}{v}
\bmdefine{\bmiwd}{w}
\bmdefine{\bmixd}{x}
\bmdefine{\bmiyd}{y}
\bmdefine{\bmizd}{z}

\bmdefine{\bmialphad}{\alpha}
\bmdefine{\bmibetad}{\beta}
\bmdefine{\bmichid}{\chi}
\bmdefine{\bmideltad}{\delta}
\bmdefine{\bmiepsilond}{\epsilon}
\bmdefine{\bmivarepsilond}{\varepsilon}
\bmdefine{\bmietad}{\eta}
\bmdefine{\bmigammad}{\gamma}
\bmdefine{\bmiiotad}{\iota}
\bmdefine{\bmikappad}{\kappa}
\bmdefine{\bmivarkappad}{\varkappa}
\bmdefine{\bmilambdad}{\lambda}
\bmdefine{\bmimud}{\mu}
\bmdefine{\bminud}{\nu}
\bmdefine{\bmiomegad}{\omega}
\bmdefine{\bmiphid}{\phi}
\bmdefine{\bmivarphid}{\varphi}
\bmdefine{\bmipid}{\pi}
\bmdefine{\bmivarpid}{\varpi}
\bmdefine{\bmipsid}{\psi}
\bmdefine{\bmirhod}{\rho}
\bmdefine{\bmivarrhod}{\varrho}
\bmdefine{\bmisigmad}{\sigma}
\bmdefine{\bmivarsigmad}{\varsigma}
\bmdefine{\bmitaud}{\tau}
\bmdefine{\bmithetad}{\theta}
\bmdefine{\bmivarthetad}{\vartheta}
\bmdefine{\bmiupsilond}{\upsilon}
\bmdefine{\bmixid}{\xi}
\bmdefine{\bmizetad}{\zeta}

\safemath{\bmia}{\bmiad}
\safemath{\bmib}{\bmibd}
\safemath{\bmic}{\bmicd}
\safemath{\bmid}{\bmidd}
\safemath{\bmie}{\bmied}
\safemath{\bmif}{\bmifd}
\safemath{\bmig}{\bmigd}
\safemath{\bmih}{\bmihd}
\safemath{\bmii}{\bmiid}
\safemath{\bmij}{\bmijd}
\safemath{\bmik}{\bmikd}
\safemath{\bmil}{\bmild}
\safemath{\bmim}{\bmimd}
\safemath{\bmin}{\bmind}
\safemath{\bmio}{\bmiod}
\safemath{\bmip}{\bmipd}
\safemath{\bmiq}{\bmiqd}
\safemath{\bmir}{\bmird}
\safemath{\bmis}{\bmisd}
\safemath{\bmit}{\bmitd}
\safemath{\bmiu}{\bmiud}
\safemath{\bmiv}{\bmivd}
\safemath{\bmiw}{\bmiwd}
\safemath{\bmix}{\bmixd}
\safemath{\bmiy}{\bmiyd}
\safemath{\bmiz}{\bmizd}

\safemath{\bmialpha}{\bmialphad}
\safemath{\bmibeta}{\bmibetad}
\safemath{\bmichi}{\bmichid}
\safemath{\bmidelta}{\bmideltad}
\safemath{\bmiepsilon}{\bmiepsilond}
\safemath{\bmivarepsilon}{\bmivarepsilond}
\safemath{\bmieta}{\bmietad}
\safemath{\bmigamma}{\bmigammad}
\safemath{\bmiiota}{\bmiiotad}
\safemath{\bmikappa}{\bmikappad}
\safemath{\bmivarkappa}{\bmivarkappad}
\safemath{\bmilambda}{\bmilambdad}
\safemath{\bmimu}{\bmimud}
\safemath{\bminu}{\bminud}
\safemath{\bmiomega}{\bmiomegad}
\safemath{\bmiphi}{\bmiphid}
\safemath{\bmivarphi}{\bmivarphid}
\safemath{\bmipi}{\bmipid}
\safemath{\bmivarpi}{\bmivarpid}
\safemath{\bmipsi}{\bmipsid}
\safemath{\bmirho}{\bmirhod}
\safemath{\bmivarrho}{\bmivarrhod}
\safemath{\bmisigma}{\bmisigmad}
\safemath{\bmivarsigma}{\bmivarsigmad}
\safemath{\bmitau}{\bmitaud}
\safemath{\bmitheta}{\bmithetad}
\safemath{\bmivartheta}{\bmivarthetad}
\safemath{\bmiupsilon}{\bmiupsilond}
\safemath{\bmixi}{\bmixid}
\safemath{\bmizeta}{\bmizetad}

\bmdefine{\bmuDeltad}{\Updelta}
\bmdefine{\bmuGammad}{\Upgamma}
\bmdefine{\bmuLambdad}{\Uplambda}
\bmdefine{\bmuOmegad}{\Upomega}
\bmdefine{\bmuPhid}{\Upphi}
\bmdefine{\bmuPid}{\Uppi}
\bmdefine{\bmuPsid}{\Uppsi}
\bmdefine{\bmuSigmad}{\Upsigma}
\bmdefine{\bmuThetad}{\Uptheta}
\bmdefine{\bmuUpsilond}{\Upupsilon}
\bmdefine{\bmuXid}{\Upxi}

\safemath{\bmuA}{\mathbf{A}}
\safemath{\bmuB}{\mathbf{B}}
\safemath{\bmuC}{\mathbf{C}}
\safemath{\bmuD}{\mathbf{D}}
\safemath{\bmuE}{\mathbf{E}}
\safemath{\bmuF}{\mathbf{F}}
\safemath{\bmuG}{\mathbf{G}}
\safemath{\bmuH}{\mathbf{H}}
\safemath{\bmuI}{\mathbf{I}}
\safemath{\bmuJ}{\mathbf{J}}
\safemath{\bmuK}{\mathbf{K}}
\safemath{\bmuL}{\mathbf{L}}
\safemath{\bmuM}{\mathbf{M}}
\safemath{\bmuN}{\mathbf{N}}
\safemath{\bmuO}{\mathbf{O}}
\safemath{\bmuP}{\mathbf{P}}
\safemath{\bmuQ}{\mathbf{Q}}
\safemath{\bmuR}{\mathbf{R}}
\safemath{\bmuS}{\mathbf{S}}
\safemath{\bmuT}{\mathbf{T}}
\safemath{\bmuU}{\mathbf{U}}
\safemath{\bmuV}{\mathbf{V}}
\safemath{\bmuW}{\mathbf{W}}
\safemath{\bmuX}{\mathbf{X}}
\safemath{\bmuY}{\mathbf{Y}}
\safemath{\bmuZ}{\mathbf{Z}}

\safemath{\bmuZero}{\mathbf{0}}
\safemath{\bmuOne}{\mathbf{1}}

\safemath{\bmuDelta}{\bmuDeltad}
\safemath{\bmuGamma}{\bmuGammad}
\safemath{\bmuLambda}{\bmuLambdad}
\safemath{\bmuOmega}{\bmuOmegad}
\safemath{\bmuPhi}{\bmuPhid}
\safemath{\bmuPi}{\bmuPid}
\safemath{\bmuPsi}{\bmuPsid}
\safemath{\bmuSigma}{\bmuSigmad}
\safemath{\bmuTheta}{\bmuThetad}
\safemath{\bmuUpsilon}{\bmuUpsilond}
\safemath{\bmuXi}{\bmuXid}

\bmdefine{\bmiAd}{A}
\bmdefine{\bmiBd}{B}
\bmdefine{\bmiCd}{C}
\bmdefine{\bmiDd}{D}
\bmdefine{\bmiEd}{E}
\bmdefine{\bmiFd}{F}
\bmdefine{\bmiGd}{G}
\bmdefine{\bmiHd}{H}
\bmdefine{\bmiId}{I}
\bmdefine{\bmiJd}{J}
\bmdefine{\bmiKd}{K}
\bmdefine{\bmiLd}{L}
\bmdefine{\bmiMd}{M}
\bmdefine{\bmiOd}{N}
\bmdefine{\bmiPd}{O}
\bmdefine{\bmiQd}{P}
\bmdefine{\bmiRd}{R}
\bmdefine{\bmiSd}{S}
\bmdefine{\bmiTd}{T}
\bmdefine{\bmiUd}{U}
\bmdefine{\bmiVd}{V}
\bmdefine{\bmiWd}{W}
\bmdefine{\bmiXd}{X}
\bmdefine{\bmiYd}{Y}
\bmdefine{\bmiZd}{Z}

\bmdefine{\bmiDeltad}{\Delta}
\bmdefine{\bmiGammad}{\Gamma}
\bmdefine{\bmiLambdad}{\Lambda}
\bmdefine{\bmiOmegad}{\Omega}
\bmdefine{\bmiPhid}{\Phi}
\bmdefine{\bmiPid}{\Pi}
\bmdefine{\bmiPsid}{\Psi}
\bmdefine{\bmiSigmad}{\Sigma}
\bmdefine{\bmiThetad}{\Theta}
\bmdefine{\bmiUpsilond}{\Upsilon}
\bmdefine{\bmiXid}{\Xi}

\safemath{\bmiA}{\bmiAd}
\safemath{\bmiB}{\bmiBd}
\safemath{\bmiC}{\bmiCd}
\safemath{\bmiD}{\bmiDd}
\safemath{\bmiE}{\bmiEd}
\safemath{\bmiF}{\bmiFd}
\safemath{\bmiG}{\bmiGd}
\safemath{\bmiH}{\bmiHd}
\safemath{\bmiI}{\bmiId}
\safemath{\bmiJ}{\bmiJd}
\safemath{\bmiK}{\bmiKd}
\safemath{\bmiL}{\bmiLd}
\safemath{\bmiM}{\bmiMd}
\safemath{\bmiN}{\bmiNd}
\safemath{\bmiO}{\bmiOd}
\safemath{\bmiP}{\bmiPd}
\safemath{\bmiQ}{\bmiQd}
\safemath{\bmiR}{\bmiRd}
\safemath{\bmiS}{\bmiSd}
\safemath{\bmiT}{\bmiTd}
\safemath{\bmiU}{\bmiUd}
\safemath{\bmiV}{\bmiVd}
\safemath{\bmiW}{\bmiWd}
\safemath{\bmiX}{\bmiXd}
\safemath{\bmiY}{\bmiYd}
\safemath{\bmiZ}{\bmiZd}

\safemath{\bmiDelta}{\bmiDeltad}
\safemath{\bmiGamma}{\bmiGammad}
\safemath{\bmiLambda}{\bmiLambdad}
\safemath{\bmiOmega}{\bmiOmegad}
\safemath{\bmiPhi}{\bmiPhid}
\safemath{\bmiPi}{\bmiPid}
\safemath{\bmiPsi}{\bmiPsid}
\safemath{\bmiSigma}{\bmiSigmad}
\safemath{\bmiTheta}{\bmiThetad}
\safemath{\bmiUpsilon}{\bmiUpsilond}
\safemath{\bmiXi}{\bmiXid}

\safemath{\evA}{\mathcal{A}}
\safemath{\evB}{\mathcal{B}}
\safemath{\evC}{\mathcal{C}}
\safemath{\evD}{\mathcal{D}}
\safemath{\evE}{\mathcal{E}}
\safemath{\evF}{\mathcal{F}}
\safemath{\evG}{\mathcal{G}}
\safemath{\evH}{\mathcal{H}}
\safemath{\evI}{\mathcal{I}}
\safemath{\evJ}{\mathcal{J}}
\safemath{\evK}{\mathcal{K}}
\safemath{\evL}{\mathcal{L}}
\safemath{\evM}{\mathcal{M}}
\safemath{\evN}{\mathcal{N}}
\safemath{\evO}{\mathcal{O}}
\safemath{\evP}{\mathcal{P}}
\safemath{\evQ}{\mathcal{Q}}
\safemath{\evR}{\mathcal{R}}
\safemath{\evS}{\mathcal{S}}
\safemath{\evT}{\mathcal{T}}
\safemath{\evU}{\mathcal{U}}
\safemath{\evV}{\mathcal{V}}
\safemath{\evW}{\mathcal{W}}
\safemath{\evX}{\mathcal{X}}
\safemath{\evY}{\mathcal{Y}}
\safemath{\evZ}{\mathcal{Z}}

\safemath{\setA}{\mathcal{A}}
\safemath{\setB}{\mathcal{B}}
\safemath{\setC}{\mathcal{C}}
\safemath{\setD}{\mathcal{D}}
\safemath{\setE}{\mathcal{E}}
\safemath{\setF}{\mathcal{F}}
\safemath{\setG}{\mathcal{G}}
\safemath{\setH}{\mathcal{H}}
\safemath{\setI}{\mathcal{I}}
\safemath{\setJ}{\mathcal{J}}
\safemath{\setK}{\mathcal{K}}
\safemath{\setL}{\mathcal{L}}
\safemath{\setM}{\mathcal{M}}
\safemath{\setN}{\mathcal{N}}
\safemath{\setO}{\mathcal{O}}
\safemath{\setP}{\mathcal{P}}
\safemath{\setQ}{\mathcal{Q}}
\safemath{\setR}{\mathcal{R}}
\safemath{\setS}{\mathcal{S}}
\safemath{\setT}{\mathcal{T}}
\safemath{\setU}{\mathcal{U}}
\safemath{\setV}{\mathcal{V}}
\safemath{\setW}{\mathcal{W}}
\safemath{\setX}{\mathcal{X}}
\safemath{\setY}{\mathcal{Y}}
\safemath{\setZ}{\mathcal{Z}}
\safemath{\emptySet}{\varnothing}

\safemath{\colA}{\mathscr{A}}
\safemath{\colB}{\mathscr{B}}
\safemath{\colC}{\mathscr{C}}
\safemath{\colD}{\mathscr{D}}
\safemath{\colE}{\mathscr{E}}
\safemath{\colF}{\mathscr{F}}
\safemath{\colG}{\mathscr{G}}
\safemath{\colH}{\mathscr{H}}
\safemath{\colI}{\mathscr{I}}
\safemath{\colJ}{\mathscr{J}}
\safemath{\colK}{\mathscr{K}}
\safemath{\colL}{\mathscr{L}}
\safemath{\colM}{\mathscr{M}}
\safemath{\colN}{\mathscr{N}}
\safemath{\colO}{\mathscr{O}}
\safemath{\colP}{\mathscr{P}}
\safemath{\colQ}{\mathscr{Q}}
\safemath{\colR}{\mathscr{R}}
\safemath{\colS}{\mathscr{S}}
\safemath{\colT}{\mathscr{T}}
\safemath{\colU}{\mathscr{U}}
\safemath{\colV}{\mathscr{V}}
\safemath{\colW}{\mathscr{W}}
\safemath{\colX}{\mathscr{X}}
\safemath{\colY}{\mathscr{Y}}
\safemath{\colZ}{\mathscr{Z}}

\safemath{\opA}{\mathbb{A}}
\safemath{\opB}{\mathbb{B}}
\safemath{\opC}{\mathbb{C}}
\safemath{\opD}{\mathbb{D}}
\safemath{\opE}{\mathbb{E}}
\safemath{\opF}{\mathbb{F}}
\safemath{\opG}{\mathbb{G}}
\safemath{\opH}{\mathbb{H}}
\safemath{\opI}{\mathbb{I}}
\safemath{\opJ}{\mathbb{J}}
\safemath{\opK}{\mathbb{K}}
\safemath{\opL}{\mathbb{L}}
\safemath{\opM}{\mathbb{M}}
\safemath{\opN}{\mathbb{N}}
\safemath{\opO}{\mathbb{O}}
\safemath{\opP}{\mathbb{P}}
\safemath{\opQ}{\mathbb{Q}}
\safemath{\opR}{\mathbb{R}}
\safemath{\opS}{\mathbb{S}}
\safemath{\opT}{\mathbb{T}}
\safemath{\opU}{\mathbb{U}}
\safemath{\opV}{\mathbb{V}}
\safemath{\opW}{\mathbb{W}}
\safemath{\opX}{\mathbb{X}}
\safemath{\opY}{\mathbb{Y}}
\safemath{\opZ}{\mathbb{Z}}
\safemath{\opZero}{\mathbb{O}}
\safemath{\identityop}{\opI}

\safemath{\sca}{a}
\safemath{\scb}{b}
\safemath{\scc}{c}
\safemath{\scd}{d}
\safemath{\sce}{e}
\safemath{\scf}{f}
\safemath{\scg}{g}
\safemath{\sch}{h}
\safemath{\sci}{i}
\safemath{\scj}{j}
\safemath{\sck}{k}
\safemath{\scl}{l}
\safemath{\scm}{m}
\safemath{\scn}{n}
\safemath{\sco}{o}
\safemath{\scp}{p}
\safemath{\scq}{q}
\safemath{\scr}{r}
\safemath{\scs}{s}
\safemath{\sct}{t}
\safemath{\scu}{u}
\safemath{\scv}{v}
\safemath{\scw}{w}
\safemath{\scx}{x}
\safemath{\scy}{y}
\safemath{\scz}{z}

\safemath{\scA}{A}
\safemath{\scB}{B}
\safemath{\scC}{C}
\safemath{\scD}{D}
\safemath{\scE}{E}
\safemath{\scF}{F}
\safemath{\scG}{G}
\safemath{\scH}{H}
\safemath{\scI}{I}
\safemath{\scJ}{J}
\safemath{\scK}{K}
\safemath{\scL}{L}
\safemath{\scM}{M}
\safemath{\scN}{N}
\safemath{\scO}{O}
\safemath{\scP}{P}
\safemath{\scQ}{Q}
\safemath{\scR}{R}
\safemath{\scS}{S}
\safemath{\scT}{T}
\safemath{\scU}{U}
\safemath{\scV}{V}
\safemath{\scW}{W}
\safemath{\scX}{X}
\safemath{\scY}{Y}
\safemath{\scZ}{Z}

\safemath{\scalpha}{\alpha}
\safemath{\scbeta}{\beta}
\safemath{\scchi}{\chi}
\safemath{\scdelta}{\delta}
\safemath{\scepsilon}{\epsilon}
\safemath{\scvarepsilon}{\varepsilon}
\safemath{\sceta}{\eta}
\safemath{\scgamma}{\gamma}
\safemath{\sciota}{\iota}
\safemath{\sckappa}{\kappa}
\safemath{\scvarkappa}{\varkappa}
\safemath{\sclambda}{\lambda}
\safemath{\scmu}{\mu}
\safemath{\scnu}{\nu}
\safemath{\scomega}{\omega}
\safemath{\scphi}{\phi}
\safemath{\scvarphi}{\varphi}
\safemath{\scpi}{\pi}
\safemath{\scvarpi}{\varpi}
\safemath{\scpsi}{\psi}
\safemath{\scrho}{\rho}
\safemath{\scvarrho}{\varrho}
\safemath{\scsigma}{\sigma}
\safemath{\scvarsigma}{\varsigma}
\safemath{\sctau}{\tau}
\safemath{\sctheta}{\theta}
\safemath{\scvartheta}{\vartheta}
\safemath{\scupsilon}{\upsilon}
\safemath{\scxi}{\xi}
\safemath{\sczeta}{\zeta}

\safemath{\veca}{\mathrm{a}}
\safemath{\vecb}{\mathrm{b}}
\safemath{\vecc}{\mathrm{c}}
\safemath{\vecd}{\mathrm{d}}
\safemath{\vece}{\mathrm{e}}
\safemath{\vecf}{\mathrm{f}}
\safemath{\vecg}{\mathrm{g}}
\safemath{\vech}{\mathrm{h}}
\safemath{\veci}{\mathrm{i}}
\safemath{\vecj}{\mathrm{j}}
\safemath{\veck}{\mathrm{k}}
\safemath{\vecl}{\mathrm{l}}
\safemath{\vecm}{\mathrm{m}}
\safemath{\vecn}{\mathrm{n}}
\safemath{\veco}{\mathrm{o}}
\safemath{\vecp}{\mathrm{p}}
\safemath{\vecq}{\mathrm{q}}
\safemath{\vecr}{\mathrm{r}}
\safemath{\vecs}{\mathrm{s}}
\safemath{\vect}{\mathrm{t}}
\safemath{\vecu}{\mathrm{u}}
\safemath{\vecv}{\mathrm{v}}
\safemath{\vecw}{\mathrm{w}}
\safemath{\vecx}{\mathrm{x}}
\safemath{\vecy}{\mathrm{y}}
\safemath{\vecz}{\mathrm{z}}
\safemath{\veczero}{\mathrm{0}}
\safemath{\vecone}{\mathrm{1}}

\safemath{\vecalpha}{\upalpha}
\safemath{\vecbeta}{\upbeta}
\safemath{\vecchi}{\upchi}
\safemath{\vecdelta}{\updelta}
\safemath{\vecepsilon}{\upepsilon}
\safemath{\vecvarepsilon}{\upvarepsilon}
\safemath{\veceta}{\upeta}
\safemath{\vecgamma}{\upgamma}
\safemath{\veciota}{\upiota}
\safemath{\veckappa}{\upkappa}
\safemath{\veclambda}{\uplambda}
\safemath{\vecmu}{\text{\textmu}}
\safemath{\vecnu}{\upnu}
\safemath{\vecomega}{\upomega}
\safemath{\vecphi}{\upphi}
\safemath{\vecvarphi}{\upvarphi}
\safemath{\vecpi}{\uppi}
\safemath{\vecvarpi}{\upvarpi}
\safemath{\vecpsi}{\uppsi}
\safemath{\vecrho}{\uprho}
\safemath{\vecvarrho}{\upvarrho}
\safemath{\vecsigma}{\upsigma}
\safemath{\vecvarsigma}{\upvarsigma}
\safemath{\vectau}{\uptau}
\safemath{\vectheta}{\uptheta}
\safemath{\vecvartheta}{\upvartheta}
\safemath{\vecupsilon}{\upupsilon}
\safemath{\vecxi}{\upxi}
\safemath{\veczeta}{\upzeta}

\safemath{\vecac}{a}
\safemath{\vecbc}{b}
\safemath{\veccc}{c}
\safemath{\vecdc}{d}
\safemath{\vecec}{e}
\safemath{\vecfc}{f}
\safemath{\vecgc}{g}
\safemath{\vechc}{h}
\safemath{\vecic}{i}
\safemath{\vecjc}{j}
\safemath{\veckc}{k}
\safemath{\veclc}{l}
\safemath{\vecmc}{m}
\safemath{\vecnc}{n}
\safemath{\vecoc}{o}
\safemath{\vecpc}{p}
\safemath{\vecqc}{q}
\safemath{\vecrc}{r}
\safemath{\vecsc}{s}
\safemath{\vectc}{t}
\safemath{\vecuc}{u}
\safemath{\vecvc}{v}
\safemath{\vecwc}{w}
\safemath{\vecxc}{x}
\safemath{\vecyc}{y}
\safemath{\veczc}{z}

\safemath{\matA}{\mathrm{A}}
\safemath{\matB}{\mathrm{B}}
\safemath{\matC}{\mathrm{C}}
\safemath{\matD}{\mathrm{D}}
\safemath{\matE}{\mathrm{E}}
\safemath{\matF}{\mathrm{F}}
\safemath{\matG}{\mathrm{G}}
\safemath{\matH}{\mathrm{H}}
\safemath{\matI}{\mathrm{I}}
\safemath{\matJ}{\mathrm{J}}
\safemath{\matK}{\mathrm{K}}
\safemath{\matL}{\mathrm{L}}
\safemath{\matM}{\mathrm{M}}
\safemath{\matN}{\mathrm{N}}
\safemath{\matO}{\mathrm{O}}
\safemath{\matP}{\mathrm{P}}
\safemath{\matQ}{\mathrm{Q}}
\safemath{\matR}{\mathrm{R}}
\safemath{\matS}{\mathrm{S}}
\safemath{\matT}{\mathrm{T}}
\safemath{\matU}{\mathrm{U}}
\safemath{\matV}{\mathrm{V}}
\safemath{\matW}{\mathrm{W}}
\safemath{\matX}{\mathrm{X}}
\safemath{\matY}{\mathrm{Y}}
\safemath{\matZ}{\mathrm{Z}}
\safemath{\matzero}{\mathrm{0}}

\safemath{\matDelta}{\Updelta}
\safemath{\matGamma}{\Upgammma}
\safemath{\matLambda}{\Uplambda}
\safemath{\matOmega}{\Upomega}
\safemath{\matPhi}{\Upphi}
\safemath{\matPi}{\Uppi}
\safemath{\matPsi}{\Uppsi}
\safemath{\matSigma}{\Upsigma}
\safemath{\matTheta}{\Uptheta}
\safemath{\matUpsilon}{\Upupsilon}
\safemath{\matXi}{\Upxi}

\safemath{\matidentity}{\matI}
\safemath{\vecunit}{\vece} % i-th unit vector
\safemath{\matone}{\matO}

\safemath{\matAc}{a}
\safemath{\matBc}{b}
\safemath{\matCc}{c}
\safemath{\matDc}{d}
\safemath{\matEc}{e}
\safemath{\matFc}{f}
\safemath{\matGc}{g}
\safemath{\matHc}{h}
\safemath{\matIc}{i}
\safemath{\matJc}{j}
\safemath{\matKc}{k}
\safemath{\matLc}{l}
\safemath{\matMc}{m}
\safemath{\matNc}{n}
\safemath{\matOc}{o}
\safemath{\matPc}{p}
\safemath{\matQc}{q}
\safemath{\matRc}{r}
\safemath{\matSc}{s}
\safemath{\matTc}{t}
\safemath{\matUc}{u}
\safemath{\matVc}{v}
\safemath{\matWc}{w}
\safemath{\matXc}{x}
\safemath{\matYc}{y}
\safemath{\matZc}{z}

\safemath{\rnda}{\mathsf{a}}
\safemath{\rndb}{\mathsf{b}}
\safemath{\rndc}{\mathsf{c}}
\safemath{\rndd}{\mathsf{d}}
\safemath{\rnde}{\mathsf{e}}
\safemath{\rndf}{\mathsf{f}}
\safemath{\rndg}{\mathsf{g}}
\safemath{\rndh}{\mathsf{h}}
\safemath{\rndi}{\mathsf{i}}
\safemath{\rndj}{\mathsf{j}}
\safemath{\rndk}{\mathsf{k}}
\safemath{\rndl}{\mathsf{l}}
\safemath{\rndm}{\mathsf{m}}
\safemath{\rndn}{\mathsf{n}}
\safemath{\rndo}{\mathsf{o}}
\safemath{\rndp}{\mathsf{p}}
\safemath{\rndq}{\mathsf{q}}
\safemath{\rndr}{\mathsf{r}}
\safemath{\rnds}{\mathsf{s}}
\safemath{\rndt}{\mathsf{t}}
\safemath{\rndu}{\mathsf{u}}
\safemath{\rndv}{\mathsf{v}}
\safemath{\rndw}{\mathsf{w}}
\safemath{\rndx}{\mathsf{x}}
\safemath{\rndy}{\mathsf{y}}
\safemath{\rndz}{\mathsf{z}}

\safemath{\rndA}{\bmiA}
\safemath{\rndB}{\bmiB}
\safemath{\rndC}{\bmiC}
\safemath{\rndD}{\bmiD}
\safemath{\rndE}{\bmiE}
\safemath{\rndF}{\bmiF}
\safemath{\rndG}{\bmiG}
\safemath{\rndH}{\bmiH}
\safemath{\rndI}{\bmiI}
\safemath{\rndJ}{\bmiJ}
\safemath{\rndK}{\bmiK}
\safemath{\rndL}{\bmiL}
\safemath{\rndM}{\bmiM}
\safemath{\rndN}{\bmiN}
\safemath{\rndO}{\bmiO}
\safemath{\rndP}{\bmiP}
\safemath{\rndQ}{\bmiQ}
\safemath{\rndR}{\bmiR}
\safemath{\rndS}{\bmiS}
\safemath{\rndT}{\bmiT}
\safemath{\rndU}{\bmiU}
\safemath{\rndV}{\bmiV}
\safemath{\rndW}{\bmiW}
\safemath{\rndX}{\bmiX}
\safemath{\rndY}{\bmiY}
\safemath{\rndZ}{\bmiZ}

\safemath{\rndalpha}{\bmialpha}
\safemath{\rndbeta}{\bmibeta}
\safemath{\rndchi}{\bmichi}
\safemath{\rnddelta}{\bmidelta}
\safemath{\rndepsilon}{\bmiepsilon}
\safemath{\rndvarepsilon}{\bmivarepsilon}
\safemath{\rndeta}{\bmieta}
\safemath{\rndgamma}{\bmigamma}
\safemath{\rndiota}{\bmiiota}
\safemath{\rndkappa}{\bmikappa}
\safemath{\rndlambda}{\bmilambda}
\safemath{\rndmu}{\bmimu}
\safemath{\rndnu}{\bminu}
\safemath{\rndomega}{\bmiomega}
\safemath{\rndphi}{\bmiphi}
\safemath{\rndvarphi}{\bmivarphi}
\safemath{\rndpi}{\bmipi}
\safemath{\rndvarpi}{\bmivarpi}
\safemath{\rndpsi}{\bmipsi}
\safemath{\rndrho}{\bmirho}
\safemath{\rndvarrho}{\bmivarrho}
\safemath{\rndsigma}{\bmisigma}
\safemath{\rndvarsigma}{\bmivarsigma}
\safemath{\rndtau}{\bmitau}
\safemath{\rndtheta}{\bmitheta}
\safemath{\rndvartheta}{\bmivartheta}
\safemath{\rndupsilon}{\bmiupsilon}
\safemath{\rndxi}{\bmixi}
\safemath{\rndzeta}{\bmizeta}

\safemath{\rveca}{\mathbf{a}}
\safemath{\rvecb}{\mathbf{b}}
\safemath{\rvecc}{\mathbf{c}}
\safemath{\rvecd}{\mathbf{d}}
\safemath{\rvece}{\mathbf{e}}
\safemath{\rvecf}{\mathbf{f}}
\safemath{\rvecg}{\mathbf{g}}
\safemath{\rvech}{\mathbf{h}}
\safemath{\rveci}{\mathbf{i}}
\safemath{\rvecj}{\mathbf{j}}
\safemath{\rveck}{\mathbf{k}}
\safemath{\rvecl}{\mathbf{l}}
\safemath{\rvecm}{\mathbf{m}}
\safemath{\rvecn}{\mathbf{n}}
\safemath{\rveco}{\mathbf{o}}
\safemath{\rvecp}{\mathbf{p}}
\safemath{\rvecq}{\mathbf{q}}
\safemath{\rvecr}{\mathbf{r}}
\safemath{\rvecs}{\mathbf{s}}
\safemath{\rvect}{\mathbf{t}}
\safemath{\rvecu}{\mathbf{u}}
\safemath{\rvecv}{\mathbf{v}}
\safemath{\rvecw}{\mathbf{w}}
\safemath{\rvecx}{\mathbf{x}}
\safemath{\rvecy}{\mathbf{y}}
\safemath{\rvecz}{\mathbf{z}}

\safemath{\rvecalpha}{\bmualpha}
\safemath{\rvecbeta}{\bmubeta}
\safemath{\rvecchi}{\bmuchi}
\safemath{\rvecdelta}{\bmudelta}
\safemath{\rvecepsilon}{\bmuepsilon}
\safemath{\rvecvarepsilon}{\bmuvarepsilon}
\safemath{\rveceta}{\bmueta}
\safemath{\rvecgamma}{\bmugamma}
\safemath{\rveciota}{\bmuiota}
\safemath{\rveckappa}{\bmukappa}
\safemath{\rveclambda}{\bmulambda}
\safemath{\rvecmu}{\bmumu}
\safemath{\rvecnu}{\bmunu}
\safemath{\rvecomega}{\bmuomega}
\safemath{\rvecphi}{\bmuphi}
\safemath{\rvecvarphi}{\bmuvarphi}
\safemath{\rvecpi}{\bmupi}
\safemath{\rvecvarpi}{\bmuvarpi}
\safemath{\rvecpsi}{\bmupsi}
\safemath{\rvecrho}{\bmurho}
\safemath{\rvecvarrho}{\bmuvarrho}
\safemath{\rvecsigma}{\bmusigma}
\safemath{\rvecvarsigma}{\bmuvarsigma}
\safemath{\rvectau}{\bmutau}
\safemath{\rvectheta}{\bmutheta}
\safemath{\rvecvartheta}{\bmuvartheta}
\safemath{\rvecupsilon}{\bmuupsilon}
\safemath{\rvecxi}{\bmuxi}
\safemath{\rveczeta}{\bmuzeta}

\safemath{\rvecac}{\rnda}
\safemath{\rvecbc}{\rndb}
\safemath{\rveccc}{\rndc}
\safemath{\rvecdc}{\rndd}
\safemath{\rvecec}{\rnde}
\safemath{\rvecfc}{\rndf}
\safemath{\rvecgc}{\rndg}
\safemath{\rvechc}{\rndh}
\safemath{\rvecic}{\rndi}
\safemath{\rvecjc}{\rndj}
\safemath{\rveckc}{\rndk}
\safemath{\rveclc}{\rndl}
\safemath{\rvecmc}{\rndm}
\safemath{\rvecnc}{\rndn}
\safemath{\rvecoc}{\rndo}
\safemath{\rvecpc}{\rndp}
\safemath{\rvecqc}{\rndq}
\safemath{\rvecrc}{\rndr}
\safemath{\rvecsc}{\rnds}
\safemath{\rvectc}{\rndt}
\safemath{\rvecuc}{\rndu}
\safemath{\rvecvc}{\rndv}
\safemath{\rvecwc}{\rndw}
\safemath{\rvecxc}{\rndx}
\safemath{\rvecyc}{\rndy}
\safemath{\rveczc}{\rndz}

\safemath{\rmatA}{\mathbf{A}}
\safemath{\rmatB}{\mathbf{B}}
\safemath{\rmatC}{\mathbf{C}}
\safemath{\rmatD}{\mathbf{D}}
\safemath{\rmatE}{\mathbf{E}}
\safemath{\rmatF}{\mathbf{F}}
\safemath{\rmatG}{\mathbf{G}}
\safemath{\rmatH}{\mathbf{H}}
\safemath{\rmatI}{\mathbf{I}}
\safemath{\rmatJ}{\mathbf{J}}
\safemath{\rmatK}{\mathbf{K}}
\safemath{\rmatL}{\mathbf{L}}
\safemath{\rmatM}{\mathbf{M}}
\safemath{\rmatN}{\mathbf{N}}
\safemath{\rmatO}{\mathbf{O}}
\safemath{\rmatP}{\mathbf{P}}
\safemath{\rmatQ}{\mathbf{Q}}
\safemath{\rmatR}{\mathbf{R}}
\safemath{\rmatS}{\mathbf{S}}
\safemath{\rmatT}{\mathbf{T}}
\safemath{\rmatU}{\mathbf{U}}
\safemath{\rmatV}{\mathbf{V}}
\safemath{\rmatW}{\mathbf{W}}
\safemath{\rmatX}{\mathbf{X}}
\safemath{\rmatY}{\mathbf{Y}}
\safemath{\rmatZ}{\mathbf{Z}}
\safemath{\rmatDelta}{\bmuDelta}
\safemath{\rmatGamma}{\bmuGamma}
\safemath{\rmatLambda}{\bmuLambda}
\safemath{\rmatOmega}{\bmuOmega}
\safemath{\rmatPhi}{\bmuPhi}
\safemath{\rmatPi}{\bmuPi}
\safemath{\rmatPsi}{\bmuPsi}
\safemath{\rmatSigma}{\bmuSigma}
\safemath{\rmatTheta}{\bmuTheta}
\safemath{\rmatUpsilon}{\bmuUpsilon}
\safemath{\rmatXi}{\bmuXi}

\safemath{\rmatAc}{\rnda}
\safemath{\rmatBc}{\rndb}
\safemath{\rmatCc}{\rndc}
\safemath{\rmatDc}{\rndd}
\safemath{\rmatEc}{\rnde}
\safemath{\rmatFc}{\rndf}
\safemath{\rmatGc}{\rndg}
\safemath{\rmatHc}{\rndh}
\safemath{\rmatIc}{\rndi}
\safemath{\rmatJc}{\rndj}
\safemath{\rmatKc}{\rndk}
\safemath{\rmatLc}{\rndl}
\safemath{\rmatMc}{\rndm}
\safemath{\rmatNc}{\rndn}
\safemath{\rmatOc}{\rndo}
\safemath{\rmatPc}{\rndp}
\safemath{\rmatQc}{\rndq}
\safemath{\rmatRc}{\rndr}
\safemath{\rmatSc}{\rnds}
\safemath{\rmatTc}{\rndt}
\safemath{\rmatUc}{\rndu}
\safemath{\rmatVc}{\rndv}
\safemath{\rmatWc}{\rndw}
\safemath{\rmatXc}{\rndx}
\safemath{\rmatYc}{\rndy}
\safemath{\rmatZc}{\rndz}

\newenvironment{textbmatrix}{	\setlength{\arraycolsep}{2.5pt}%
								\big[\begin{matrix}}{\end{matrix}\big]%
								\raisebox{0.08ex}{\vphantom{M}}}

 \def\btm{\begin{textbmatrix}}
 \def\etm{\end{textbmatrix}}

\DeclareMathOperator{\tr}{tr}				% trace
\DeclareMathOperator{\diag}{diag}			% diagonal matrix
\DeclareMathOperator{\rank}{rank}			% rank of a matrix
\DeclareMathOperator{\adj}{adj}				% adjunct matrix
\safemath{\fun}{\scf}						% generic scalar function

\safemath{\vrbl}{x}						% generic vector variable
\safemath{\altvrbl}{y}						% alt generic vector variable
\safemath{\aaltvrbl}{z}						% aalt generic vector variable

\safemath{\vvrbl}{\vecx}						% generic vector variable
\safemath{\altvvrbl}{\vecy}						% alt generic vector variable
\safemath{\aaltvvrbl}{\vecz}						% aalt generic vector variable

\safemath{\altfun}{\scg}
\safemath{\aaltfun}{\sch}
\safemath{\bel}{\sce}					% basis element
\safemath{\altbel}{\sce}					% alternative basis element
\safemath{\frel}{g}					% frame element
\safemath{\altfrel}{g}					% alternative frame element
\safemath{\dfrel}{\tilde{g}}					% dual frame element
\safemath{\altdfrel}{\tilde{g}}					% alternative dual frame element

\safemath{\mat}{\matA}						% generic matrix
\safemath{\matc}{\matAc}						% components of a generic matrix
\safemath{\altmat}{\matB}						% alternative generic matrix
\safemath{\altmatc}{\matBc}						% alternative generic matrix

\safemath{\vectr}{\vecu}						% generic vector
\safemath{\vectrc}{\vecuc}						% components of a generic vector
\safemath{\altvectr}{\vecv}						% alternative generic vector
\safemath{\aaltvectr}{\vect}						% aalternative generic vector
\safemath{\altvectrc}{\vecvc}						% components of an alternative generic vector
\safemath{\genvar}{u}						% generic variable
\safemath{\altgenvar}{v}						% alternative generic variable 

\safemath{\rvectr}{\rvecu}						% random generic vector
\safemath{\rvectrc}{\rvecuc}						% random components of a generic vector
\safemath{\raltvectr}{\rvecv}						% random alternative generic vector
\safemath{\raaltvectr}{\rvect}						% random aalternative generic vector
\safemath{\raltvectrc}{\rvecvc}						% random components of an alternative generic vector
\safemath{\rgenvar}{\rndu}						% random generic variable
\safemath{\raltgenvar}{\rndv}						% random alternative generic variable 

\newcommand{\nullspace}{\setN}	 			% nullspace
\newcommand{\ind}[1]{\chi_{#1}}				% indicator function
\newcommand{\conj}[1]{\ensuremath{#1^{*}}} 	% conjugate 		
\newcommand{\tp}[1]{\ensuremath{#1^{\mathsf{T}}}} 		% transpose

\newcommand{\inv}[1]{\ensuremath{#1^{-1}}} 	% inverse
\safemath{\dirac}{\delta}					% Dirac delta
\safemath{\diracp}{\dirac(\time)}			% 	''	parametrized
\safemath{\krond}{\dirac}					% Kronecker delta
\safemath{\indfun}{I}						% Indicator function
\safemath{\stepfun}{u}						% step function at zero

\safemath{\upto}{\uparrow}
\safemath{\downto}{\downarrow}
\safemath{\iu}{\mathrm{i}}							% imaginary unit
\safemath{\maj}{\succ}
\newcommand{\dftmat}[1]{\matF_{#1}}			% DFT matrix
\safemath{\mdft}{\dftmat{}}					% 	''
\safemath{\runity}{\beta}					% root of unity
\safemath{\eval}{\lambda}					% eigenvalue

\safemath{\veval}{\veclambda}				% eigenvalue vector
\safemath{\littleo}{\sco}					% Landau\s little o

\let\im\undefined
\safemath{\re}{\Re}				% real part
\safemath{\im}{\Im}				% imaginary part
\safemath{\euclidspace}{\complexset}			% Euclidean space
\safemath{\confunspace}{\setC}				% space of continuous functions
\newcommand{\banachseqspace}[1]{l^{#1}}		% Banach sequence space
\safemath{\hilseqspace}{\banachseqspace{2}}	% Hilbert sequence space
\newcommand{\banachfunspace}[1]{\setL^{#1}}	% Banach function space
\safemath{\hilfunspace}{\banachfunspace{2}}	% Hilbert function space
\safemath{\hilfunspacep}{\hilfunspace(\complexset)}	% Hilbert function space parametrized
\safemath{\schwarzspace}{\setS}				% Schwarz space

\newcommand{\hadj}[1]{#1^{\star}}			% Hilbert adjoint operator

\safemath{\SNR}{\rho} 				% signal to noise ratio
\safemath{\SINR}{\text{\sc sinr}} 				% signal to interference plus noise ratio
\safemath{\No}{N_0}							% noise spectral density
\safemath{\Es}{E_s}							% energy per symbol
\safemath{\Eb}{E_b}							% energy per bit
\safemath{\EbNo}{\frac{\Eb}{\No}}
\safemath{\EsNo}{\frac{\Es}{\No}}
\safemath{\NoVar}{\variance}                 % noise variance

\let\time\undefined
\safemath{\time}{\sct}						% continuous time
\safemath{\dtime}{\sck}						% discrete time
\safemath{\delay}{\sctau}					% continuous delay
\safemath{\ddelay}{\scl}						% discrete delay
\safemath{\doppler}{\scnu}					% continuous doppler
\safemath{\ddoppler}{\scm}					% discrete doppler
\safemath{\freq}{\scf}						% frequency
\safemath{\dfreq}{\scn}						% discrete frequency
\safemath{\Dtime}{\Delta\time}
\safemath{\Dfreq}{\Delta\freq}
\safemath{\Ddtime}{\dtime}
\safemath{\Ddfreq}{\dfreq}
\safemath{\bandwidth}{\scB}
\safemath{\maxdoppler}{\doppler_{0}}			% maximum Doppler shift
\safemath{\maxdelay}{\delay_{0}}				% maximum delay
\safemath{\spread}{\Delta_{\CHop}}			% total channel spread
\DeclareMathOperator{\CHop}{\ensuremath{\opH}} % channel operator
\safemath{\kernel}{\rndk_{\CHop}}			% operator kernel
\safemath{\kernelp}{\kernel(\time,\time')}	% 	''	parametrized
\safemath{\tvir}{\rndh_{\CHop}}				% time-varying impulse response
\safemath{\tvirp}{\tvir(\time,\delay)}		%	''	parametrized
\safemath{\tvirc}{\conj{\rndh}_{\CHop}}		% 	''	parametrized
\safemath{\tvtf}{\rndl_{\CHop}}				% time-varying transfer function
\safemath{\tvtfp}{\tvtf(\time,\freq)}			%	''	parametrized
\safemath{\tvtfc}{\conj{\rndl}_{\CHop}}		%	''	parametrized
\safemath{\spf}{\rnds_{\CHop}}				% spreading function
\safemath{\spfp}{\spf(\doppler,\delay)}		%	''	parametrized
\safemath{\spfc}{\conj{\rnds}_{\CHop}}		%	''	parametrized
\safemath{\bff}{\rndb_{\CHop}}				% bi-freuqency function
\safemath{\bffp}{\bff(\doppler,\freq)}		%	''	parametrized

\safemath{\irc}{\scr_{\rndh}}				% impulse response correlation fn.
\safemath{\tfc}{\scr_{\rndl}}				% time-frequency correlation fn.
\safemath{\spc}{\scr_{\rnds}}				% spreading fn. correlation fn.
\safemath{\bfc}{\scr_{\rndb}}				% bi-frequency correlation fn.

\safemath{\scaf}{\scc_{\rnds}}				% scattering function
\safemath{\scafp}{\scaf(\doppler,\delay)}		% 	''	parametrized
\safemath{\ccf}{\scc_{\rndl}}				% WSSUS tvtf correlation
\safemath{\ccfp}{\ccf(\Dtime,\Dfreq)}			% 	''	parametrized
\safemath{\cic}{\scc_{\rndh}}				% WSSUS tvir correlation
\safemath{\cicp}{\cic(\Dtime,\delay)}			% 	''	parametrized

\safemath{\mi}{I}							% mutual information
\safemath{\capacity}{C}					% capacity

\DeclareMathOperator{\Prob}{\opP}		% probability of an event
\safemath{\normal}{\mathcal{N}}			% normal distribution
\safemath{\jpg}{\mathcal{CN}}			% jointly proper Gaussian
\safemath{\uniform}{\mathcal{U}}				% uniform distribution
\safemath{\mchain}{\leftrightarrow}		% Markov chain
\safemath{\dB}{\,\mathrm{dB}}
\safemath{\dBm}{\,\mathrm{dBm}}
\safemath{\Hz}{\,\mathrm{Hz}}
\safemath{\kHz}{\,\mathrm{kHz}}
\safemath{\MHz}{\,\mathrm{MHz}}
\safemath{\GHz}{\,\mathrm{GHz}}
\safemath{\s}{\,\mathrm{s}}
\safemath{\ms}{\,\mathrm{ms}}
\safemath{\mus}{\,\mathrm{\text{\textmu}s}}
\safemath{\ns}{\,\mathrm{ns}}
\safemath{\ps}{\,\mathrm{ps}}
\safemath{\meter}{\,\mathrm{m}}
\safemath{\mm}{\,\mathrm{mm}}
\safemath{\cm}{\,\mathrm{cm}}
\safemath{\m}{\,\mathrm{m}}
\safemath{\W}{\,\mathrm{W}}
\safemath{\mW}{\, \mathrm{mW}}
\safemath{\J}{\,\mathrm{J}}
\safemath{\K}{\,\mathrm{K}}
\safemath{\bit}{\,\mathrm{bit}}
\safemath{\nat}{\,\mathrm{nat}}

\safemath{\define}{\triangleq}					% definition

\safemath{\equivalent}{\sim}
\safemath{\distas}{\sim}					% distributed according to
\safemath{\sdiff}{\Delta}				% symmetric set difference
\safemath{\setdiff}{\setminus}				% set difference

\safemath{\reals}{\mathbb R}
\safemath{\positivereals}{\reals^{+}}
\safemath{\integers}{\mathbb Z}
\safemath{\posint}{\integers^{+}}
\safemath{\naturals}{\mathbb N}
\safemath{\posnaturals}{\naturals^{+}}
\safemath{\complexset}{\mathbb C}
\safemath{\rationals}{\mathbb Q}

\safemath{\iSet}{\setI}
\safemath{\rel}{\bowtie}					% relation
\safemath{\eqrel}{\sim}					% equivalence relation
\safemath{\rlord}{\leq}					% reflexive linear ordering
\safemath{\slord}{<}						% strict linear ordering
\safemath{\rpord}{\preceq}				% reflexive partial ordering
\safemath{\rrpord}{\succeq}				% reversed reflexive partial ordering
\safemath{\spord}{\prec}					% strict partial ordering
\safemath{\sig}{\sigma}					% sigma-{algebra, ring,...}
\safemath{\metric}{d}					% metric
\safemath{\setfun}{\Phi}					% set function
\safemath{\measure}{\mu}					% measure
\safemath{\altmeasure}{\lambda}					% measure
\newcommand{\outerm}[1]{#1^{\star}}		% marks outer measures.
\newcommand{\innerm}[1]{#1_{\star}}		% marks inner measures.
\safemath{\omeasure}{\outerm{\measure}}		% outer measure
\safemath{\imeasure}{\innerm{\measure}}		% inner measure	
\safemath{\aecol}{\colS^{\star}_{\measure}} % collection of almost equal sets
\safemath{\emeasure}{\bar{\measure}_{0}}	% measure extension
\safemath{\rmeasure}{\tilde{\measure}}	% restricted measure
\safemath{\bmeasure}{\measure_{0}}		% basic measure on a semiring
\safemath{\glength}{\measure_{\altfun}}	% generalized length
\safemath{\lebmea}{\lambda}				% Lebesgue length and measure
\safemath{\blebmea}{\lebmea_{0}}			% pre-lebesgue-measure
\safemath{\sfun}{s}						% simple function
\safemath{\absintspace}{\colL^{1}}		% space of abs. integrable functions
\safemath{\sqintspace}{\colL^{2}}		% space of square integrable functions
\safemath{\abssumspace}{l^{1}}		% space of abs. summable sequences
\safemath{\sqsumspace}{l^{2}}		% space of square summable sequences

\safemath{\sfield}{\setF}				% scalar field
\safemath{\vectors}{\setV}				% set of vectors
\safemath{\vecspace}{(\vectors,\sfield)}	% vector space
\safemath{\linspace}{\setV}				% linear space
\safemath{\clinspace}{(\linspace,\sfield)} % linear space
\safemath{\nspace}{\setU}				% normed space
\safemath{\metspace}{\setM}				% metric space
\safemath{\bspace}{\setB}				% Banach space
\safemath{\ipspace}{\setG}				% inner product space
\safemath{\hilspace}{\setH}				% Hilbert space
\safemath{\blospace}{\setG}				% set of bounded linear oprators
\safemath{\lop}{\opT}					% linear operator
\safemath{\altlop}{\opS}					% alternative linear operator
\safemath{\nullsp}{\nullspace(\lop)}		% null space of the linear operator
\safemath{\lfun}{l}						% linear functional
\safemath{\altlfun}{g}					% alternative linear functional
\newcommand{\dual}[1]{#1^{'}}			% dual space
\safemath{\dsum}{\oplus}					% direct sum
\safemath{\funspace}{\colL}				% function space
\renewcommand{\adj}[1]{#1^{\times}}		% adjoint operator generator
\safemath{\adjlop}{\adj{\lop}}			% adjoint operator
\safemath{\hadjlop}{\hadj{\lop}}			% Hilbert adjoint operator
\safemath{\tow}{\xrightarrow{w}}			% weak convergence
\safemath{\tows}{\xrightarrow{w^{*}}}		% weak* convergence
\safemath{\cparam}{\lambda}
\safemath{\lopl}{\lop_{\cparam}}		
\safemath{\iop}{\opI}					% identity operator
\safemath{\resolop}{\opR}				% resolvent operator
\safemath{\resolvent}{\resolop_{\cparam}(\lop)}	% resolvent operator
\safemath{\reset}{\setQ}
\safemath{\spectrum}{\setS}
\safemath{\resolset}{\reset(\lop)}		% resolvent set
\safemath{\lopspec}{\spectrum(\lop)}		% spectrum of a linear operator
\safemath{\pspec}{\spectrum_{p}(\lop)}	% point spectrum
\safemath{\cspec}{\spectrum_{c}(\lop)}	% continuous spectrum
\safemath{\rspec}{\spectrum_{r}(\lop)}	% residual spectrum
\safemath{\ev}{\cparam}					% eigenvalue
\newcommand{\specrad}[1]{r_{#1}}			% spectral radius
\safemath{\lopsrad}{\specrad{\lop}}		% spectral radius
\safemath{\pop}{\opP}					% projection operator

\safemath{\specfam}{\colE}				% spectral family
\safemath{\specop}{\opE_{\cparam}}		% spectral projection operator
\safemath{\altspecop}{\opE_{\mu}}		% alternat spectral projection operator
\safemath{\vmulti}{\vecone}				% vector multiplicative identity
\safemath{\unitaryop}{\opU}				% unitary operator
\safemath{\sval}{\sigma}					% singular value
\safemath{\corrcoef}{\rho}				% canonical correlation coefficient
\safemath{\sangle}{\theta}				% angle between subspaces

\let\time\undefined
\safemath{\iset}{\setI}				% index set
\safemath{\shift}{\nu}
\safemath{\scale}{\alpha}
\safemath{\time}{t}
\safemath{\specfreq}{\theta}	
\newcommand{\transopgen}[1]{\opT_{#1}} % translation operator
\safemath{\transop}{\transopgen{\delay}}
\newcommand{\modopgen}[1]{\opM_{#1}}	% modulation operator
\safemath{\modop}{\modopgen{\shift}}
\newcommand{\dilopgen}[1]{\opD_{#1}}	% dilation operator
\safemath{\dilop}{\dilopgen{\scale}}
\safemath{\fram}{\setF}				% frame
\safemath{\dfram}{\dual{\fram}}		% dual frame
\safemath{\ufb}{B}					% upper frame bound
\safemath{\lfb}{A}					% lower frame bound
\safemath{\sop}{\hadj{\aop}}				% frame synthesis operator
\safemath{\aop}{\opT}			% frame analysis operator 		// modifide by christoph bunte
\safemath{\fop}{\opS}				% frame operator 			// modified by christoph bunte 
\safemath{\daop}{\tilde\opT}			% dual frame analysis operator 		// added by christoph bunte
\safemath{\dsop}{\hadj{\tilde\opT}}				% dual frame synthesis operator 			// added by christoph bunte 

\safemath{\ifop}{\inv{\fop}}			% inverse frame operator
\safemath{\rifop}{\fop^{-1/2}}			% square root of inverse frame operator
\safemath{\transeq}{\setT}			% sequence of translates
\safemath{\nfun}{\Phi}				% ``Nyquist series''
\safemath{\funvec}{\vecf}			%  function as a vector
\safemath{\altfunvec}{\vecg}

\safemath{\samplespace}{\Omega}
\safemath{\probspace}{(\samplespace,\sfield,\Prob)}	% probability space
\safemath{\ccoef}{\rho}			% correlation coefficient
\safemath{\infstate}{\vecpi}				% steady state vector
\safemath{\typset}{\setA_{\epsilon}^{(N)}}	% typical set
\safemath{\expequal}{\doteq}				% equal to first order in the exponent
\safemath{\code}{C}						% code
\safemath{\dstringset}{\setD^{\star}}		% set of finite length D-ary strings
\safemath{\cwlength}{l}					% codeword length
\safemath{\elength}{L}					% expected codeword length
\safemath{\extension}{C^{\star}}			% code extension
\safemath{\approaches}{\rightarrow}		% i.e., x_{n} -> x
\safemath{\evnt}{\setA}					% event A
\safemath{\altevnt}{\setB}					% event B
\safemath{\rv}{\rndx}					% random variable X
\safemath{\altrv}{\rndy}					% random variable Y
\safemath{\complexrv}{\rndu}					% complex random variable U
\safemath{\altcrv}{\rndv}				% complex random variable V
\safemath{\rvec}{\rvecx}					% random vector X
\safemath{\altrvec}{\rvecy}				% random vector Y
\safemath{\crvec}{\rvecu}				% complex random vector U
\safemath{\altcrvec}{\rvecv}				% complex random vector V
\safemath{\variance}{\sigma^{2}}			% variance
\safemath{\map}{T}						% mapping
\safemath{\jacobian}{\matJ}					% jacobian
\safemath{\wvec}{\rvecw}					% white random vector
\safemath{\wrv}{\rndw}					% white noise process
\safemath{\orthmat}{\matQ}				% orthogonal matrix
\safemath{\evmat}{\matLambda}			% eigenvalue matrix (diagonal)
\safemath{\identity}{\matidentity}		% identity matrix
\safemath{\innovec}{\vecv}				% innovations vector
\safemath{\convas}{\xrightarrow{\text{a.s.}}}	% almost sure convergence
\safemath{\convr}{\xrightarrow{\text{r}}}	% convergence in r-th mean
\safemath{\convp}{\xrightarrow{\text{P}}}	% convergence in probability
\safemath{\convd}{\xrightarrow{\text{D}}}	% convergence in distribution
\safemath{\ltis}{\opL}				% LTI system
\safemath{\ir}{h}					% impulse response
\safemath{\tf}{\MakeUppercase{\ir}}	% transfer function
\newcommand*{\fancyrefparlabelprefix}{par}		% Part
\newcommand*{\fancyrefremlabelprefix}{rem}		% Remark
\newcommand*{\fancyrefchalabelprefix}{cha}		% Chapter
\newcommand*{\fancyrefapplabelprefix}{app}		% Appendix
\newcommand*{\fancyrefthmlabelprefix}{thm}		% Theorem
\newcommand*{\fancyreflemlabelprefix}{lem}		% Lemma
\newcommand*{\fancyrefcorlabelprefix}{cor}		% Corolary
\newcommand*{\fancyrefdeflabelprefix}{def}		% Definition
\newcommand*{\fancyrefproplabelprefix}{prop}		% Property

\frefformat{vario}{\fancyrefparlabelprefix}{Part~#1}
\frefformat{vario}{\fancyrefremlabelprefix}{Remark~#1}
\frefformat{vario}{\fancyrefchalabelprefix}{Chapter~#1}
\frefformat{vario}{\fancyrefseclabelprefix}{Section~#1}
\frefformat{vario}{\fancyrefthmlabelprefix}{Theorem~#1}
\frefformat{vario}{\fancyreflemlabelprefix}{Lemma~#1}
\frefformat{vario}{\fancyrefcorlabelprefix}{Corollary~#1}
\frefformat{vario}{\fancyrefdeflabelprefix}{Definition~#1}
\frefformat{vario}{\fancyreffiglabelprefix}{Figure~#1}
\frefformat{vario}{\fancyrefapplabelprefix}{Appendix~#1}
\frefformat{vario}{\fancyrefeqlabelprefix}{(#1)}
\frefformat{vario}{\fancyrefproplabelprefix}{Property~#1}

\theoremstyle{plain}
\newtheorem{thm}{Theorem}
\newtheorem{lem}{Lemma}

\theoremstyle{definition}
\newtheorem{dfn}{Definition}

\theoremstyle{remark}
\newtheorem{rem}{Remark}
\newtheorem{exa}{Example}

\renewcommand{\Pr}{\operatorname{P}}
\allowdisplaybreaks[4]
 
\renewcommand{\dots}{\!...} 
 
\usepackage{mathtools}  
\mathtoolsset{showonlyrefs=true}    
\newcommand{\lebmeasure}{\lambda}
  
\title{Information-Theoretic Limits of Matrix Completion}

\author{\IEEEauthorblockN{Erwin Riegler, David Stotz,   and  Helmut B\"olcskei\medskip}
\IEEEauthorblockA{Dept.~IT~\&~EE, ETH Zurich, Switzerland\\
Email: \{eriegler, dstotz, boelcskei\}@nari.ee.ethz.ch}
}

\begin{document}
\maketitle

\begin{abstract}
We propose an information-theoretic framework for matrix completion. The theory goes beyond the low-rank structure and applies to general matrices of ``low description complexity". Specifically, we consider random matrices $\rmatX\in\reals^{m\times n}$  of arbitrary distribution (continuous, discrete, discrete-continuous mixture, or even singular).  With $\setS\subseteq \reals^{m\times n}$ an $\varepsilon$-support set of $\rmatX$, i.e., $\Pr[\rmatX\in\setS]\geq 1-\varepsilon$, and  $\underline{\dim}_\mathrm{B}(\setS)$ denoting the lower Minkowski dimension of  $\setS$, we show that  $k>\underline{\dim}_\mathrm{B}(\setS)$ measurements of the form $\langle\matA_i,\matX\rangle$, with $\matA_i$ denoting the measurement matrices, suffice 
to recover $\rmatX$ with  probability of error at most $\varepsilon$. The result holds for Lebesgue a.a.  $\matA_i$ and does not need incoherence between the  $\matA_i$ and the unknown matrix $\rmatX$.  
We furthermore show that $k>\underline{\dim}_\mathrm{B}(\setS)$ measurements also suffice to recover the unknown matrix $\rmatX$ 
from measurements taken with rank-one  $\matA_i$, again this applies to a.a. rank-one $\matA_i$.
Rank-one  measurement matrices are attractive  as they require less storage space  than general measurement matrices and can be applied faster. Particularizing our results to the recovery of low-rank matrices, we find that $k>(m+n-r)r$ measurements are sufficient to recover matrices of rank at most $r$. 
Finally, we construct a class of rank-$r$ matrices that can be recovered with arbitrarily small probability of error from  $k<(m+n-r)r$ measurements. 
\end{abstract}

\section{Introduction} 
Matrix completion refers to the recovery of a low-rank matrix from a (small) subset of its entries or a (small) number of linear combinations of its entries. This problem arises in a wide range of applications, including quantum state tomography, face recognition, recommender systems, and sensor localization (see, e.g., \cite{capl11} and references therein).  %Motivated by these applications, the matrix completion problem  has attracted significant attention in statistics, electrical engineering, applied mathematics, and computer science. 

The formal problem statement is as follows. Suppose we have $k$ linear measurements of the $m\times n$ matrix 
$\matX$  with $\rank(\matX)\leq r$ in the form of
%The matrix to be recovered in the matrix recovery problem is a low-rank matrix $\matX\in\reals^{m\times n}$ given $k$ linear measurements with measurement vector 
\begin{align}\label{eq:smodel}
\vecy=\tp{(\langle\matA_1,\matX\rangle,\dots, \langle\matA_k,\matX\rangle)} \in\reals^{k} 
\end{align}
where $\matA_i\in \reals^{m\times n}$ denotes the measurement matrices and $\langle\cdot,\cdot\rangle$ stands for the standard trace inner product between matrices in $\reals^{m\times n}$. 
%The rank of the matrix $\matX$ is known to be lower than  or equal to a certain value $r$. 
The number of measurements $k$ is  typically much smaller than the total number of entries, $mn$, of $\matX$. Depending on the  $\matA_i$, the measurements can simply be individual entries of $\matX$ or general linear combinations thereof.

The vast literature on matrix completion, for a highly incomplete list see \cite{gr11,care09,cata10,refapa10,capl11, re11,cazh14,cazh15}, provides guarantees for the recovery of the unknown low-rank matrix $\matX$ from the  measurements $\vecy$, under various assumptions on the measurement matrices $\matA_i$ and the 
low-rank models generating $\matX$. For example,  in \cite{gr11} the  $\matA_i$ are assumed to be chosen randomly  from an orthonormal basis for $\reals^{n\times n}$ and it is shown that an unknown $n\times n$ matrix $\matX$ of rank at most $r$ can be recovered with high probability 
%(with respect to the choice of the measurement matrices) 
if $k=\mathcal{O}(nr\nu\ln^2 n)$. Here, $\nu$ quantifies the  incoherence between the unknown matrix $\matX$ and the orthonormal basis for $\reals^{n\times n}$ the $\matA_i$ are drawn from. 

The setting in \cite{care09} assumes random measurement matrices $\rmatA_i$  with the position of the only nonzero entry chosen uniformly at random.  It is shown that almost all (a.a.)  matrices (with respect to the random orthogonal model \cite[Def. 2.1]{care09}) of rank at most $r$ can be recovered with high probability (with respect to the measurement matrices) provided that the number of measurements satisfies $k\geq Cn^{1.25}r\ln n$, where $C$ is a numerical constant. 

In  \cite{capl11} it is shown that for measurement matrices $\rmatA_i$ containing i.i.d. 
entries (that are, e.g., Gaussian), a matrix $\matX$ of rank at most $r$ can be recovered with high probability from $k\geq C(m+n)r$ measurements, where $C$ is a constant. The recovery guarantees  in  \cite{gr11,care09,cata10,refapa10,capl11, re11,cazh14,cazh15} all pertain to recovery through nuclear norm minimization. In \cite{elnepl12}  measurement matrices $\rmatA_i$ containing i.i.d.  entries drawn from an absolutely continuous (with respect to Lebesgue measure) distribution are considered.  It is shown that rank minimization (which is NP-hard, in general) recovers an $n\times n$ matrix $\matX$ of rank at most $r$ with probability one if $k> (2n-r)r$. 
It is furthermore shown in \cite{elnepl12} that all matrices $\matX$ of rank at most $n/2$ can be recovered, again  with probability one, provided that  $k\geq 4nr-4r^2$.  
%In both cases,  recovery can be accomplished by minimizing the rank of the unknown matrix $\matX$. 
The recovery thresholds in  \cite{capl11,elnepl12} do not exhibit a $\log n$ term,
%the size of the unknown matrices (without any $\log n$ factor) 
but assume significant  richness in the random measurement matrices $\rmatA_i$. 
%One problem with this approach is that it requires a large amount of storage space for the  measurement matrices.  
%From a practical point of view, 
%rank one measurement matrices are important because they require much less storage space than rich measurement matrices.
Storing and applying such measurement matrices is costly in terms of memory and computation time.  
To overcome this problem \cite{cazh15} considers rank-one measurement matrices of the form $\rmatA_i=\rveca_i\tp{\rvecb_i}$, where
%\footnote{Here, $\veca_i\in\reals^m$ means that each realization of the random vector $\veca_i$ is an element of $\reals^m$.} 
$\rveca_i\in\reals^m$
 and $\rvecb_i \in\reals^n$ are independent with i.i.d.  Gaussian or sub-Gaussian entries, and shows that nuclear norm minimization succeeds under the same recovery threshold  as in \cite{capl11}, namely $k\geq C(m+n)r$.
%and derives the same recovery guarantees as in  which again can be achieved by minimizing the nuclear norm of the unknown matrix X. 

\emph{Contributions:} Inspired by the  work of Wu and Verd\'{u} on analog signal compression \cite{wuve10}, we formulate an  information-theoretic framework for almost lossless matrix completion. The theory is general in the sense of going beyond the low-rank structure and applying to general matrices of ``low description complexity". Specifically, we consider random matrices $\rmatX\in\reals^{m\times n}$  of arbitrary distribution (continuous, discrete, discrete-continuous mixture, or even singular).  With $\setS\subseteq \reals^{m\times n}$ an $\varepsilon$-support set of $\rmatX$, i.e., $\Pr[\rmatX\in\setS]\geq 1-\varepsilon$, and  $\underline{\dim}_\mathrm{B}(\setS)$ denoting the lower Minkowski dimension (see Definition \ref{dfndim}) of  $\setS$, we show that  $k>\underline{\dim}_\mathrm{B}(\setS)$ measurements suffice to recover $\rmatX$ with probability of error no more than  $\varepsilon$. The result holds for Lebesgue a.a. measurement matrices $\matA_i$ and does not need any incoherence between the  $\matA_i$ and the unknown matrix $\rmatX$.  What is more, we show that $k>\underline{\dim}_\mathrm{B}(\setS)$ measurements also suffice for recovery from measurements taken with rank-one $\matA_i$, again this applies to a.a. rank-one $\matA_i$.
%We then  strengthen this result and show that the unknown matrix with a.a. rank-one measurement matrices $\matA_i$. %as in \cite{cazh15}. 
%To show this result, we have to develop a concentration of measure property that is valid for rank one measurement matrices and differs from the one that we use for rich measurement matrices. 

Particularizing our results to low-rank matrices $\rmatX$, we show that 
%any bounded subset of the set of  $m\times n$ matrices of rank at most $r$ has lower Minkowski dimension less than or equal to $(m+n-r)r$, 
%, which is the Euclidean dimension of the manifold of matrices of rank $r$ ,
%which than implies that 
  $\rmatX$ of rank at most $r$ can be recovered from $k>(m+n-r)r$ measurements taken with either general $\matA_i$ or with rank-one $\matA_i$.  
%and  arbitrarily small probability of error, irrespective of the distribution of $\rmatX$. 
Perhaps surprisingly, it turns out that, depending on the specific distribution of the low-rank matrix $\rmatX$, even fewer than $(m+n-r)r$ measurements can suffice. We construct a class of examples that illuminates this phenomenon. 
%The main technical results developed in the paper are a null space property for bounded sets $\setS\subset \reals^{m\times n}$ and subspaces with the sum of the lower Minkowski dimension of $\setS$ and the dimension of the subspace being smaller than the dimension of the ambient space and a concentration of measure result for low-rank matrices. 

\emph{Notation:} 
Roman letters $\matA,\matB,\dots$  designate deterministic matrices and $\veca,\vecb,\dots$ stands for deterministic vectors. 
Boldface letters $\rmatA,\rmatB,\dots$ and $\rveca,\rvecb,\dots$ denote random matrices and  vectors, respectively. 
For the distribution of a random matrix $\rmatA$ we write $\mu_\rmatA$ and we use $\mu_\rveca$ to designate the distribution of a random vector $\rveca$. 
$\lebmeasure^{k}$ denotes the Lebesgue measure on $\reals^k$ 
The superscript  $\tp{}$ stands for transposition. 
For $\matA=(\veca_1,\dots,\veca_n)\in \reals^{m\times n}$ we let $\operatorname{vec}(\matA)=\tp{(\tp{\veca_1},\dots,\tp{\veca_n})}$.  
% identity matrix of suitable dimension is denoted by $\matI$. 
For a rank-$r$ matrix $\matA\in\reals^{m\times n}$  with ordered singular values $\sigma_1(\matA)\geq\dots\geq\sigma_r(\matA)$, we set $\Delta(\matA)=\prod_{i=1}^r\sigma_i(\matA)$. 
%and $\|\matA\|=\sigma_1(\matA)$  its spectral norm.  
For a matrix $\matA$, $\tr(\matA)$ denotes its trace. 
%and $\|\matA\|_2=\sqrt{\langle\matA,\matA\rangle}$  its Euclidean norm.
For matrices $\matA,\matB$ of the same dimensions, $\langle\matA,\matB\rangle=\tr(\tp{\matA}\matB)$ is the trace inner product between $\matA$ and $\matB$. We write $\|\matA\|_2=\sqrt{\langle\matA,\matA\rangle}$ for the Euclidean norm of the matrix $\matA$. For the Euclidean space $(\reals^k,\|\cdot\|_2)$, we denote the open ball of radius $s$ centered at $\vecu\in \reals^k$ by 
$\setB_k(\vecu,s)$, $V(k,s)$ and $A(k-1,s)$ stand for its volume and the area of its closure, respectively. 
Similarly, for the Euclidean space $(\reals^{m\times n},\|\cdot\|_2)$, we denote the open ball of radius $s$  centered at $\matA\in \reals^{m\times n}$ by  $\setB_{m\times n}(\matA,s)$. 
We write $\setM_r^{m\times n}$ and $\setN_r^{m\times n}$ for the set of matrices $\matA\in\reals^{m\times n}$ with $\rank(\matA)\leq r$ and $\rank(\matA)= r$, respectively. 
%The indicator function on a set $\setU$ is denoted by $\ind{\setU}$.

%\vspace*{-1truemm}

\section{Almost lossless matrix completion}
We start by formulating the almost lossless matrix completion framework. 
\begin{dfn}
For a random matrix $\rmatX\in\reals^{m\times n}$ of arbitrary distribution $\mu_\rmatX$ with Lebesgue decompositon 
$\mu_\rmatX=\mu^\mathrm{c}_\rmatX+\mu^\mathrm{d}_\rmatX+\mu^\mathrm{s}_\rmatX$ (continuous, discrete, and singular components, respectively), an $(m\times n,k)$ code consists of 
\begin{enumerate}[(i)]
\item linear measurements $\tp{(\langle\matA_1,\cdot \rangle,\dots, \langle\matA_k,\cdot\rangle)}:\reals^{m\times n}\to \reals^k$;
\item a measurable decoder $g:\reals^k\to \reals^{m\times n}$.
\end{enumerate}
For given measurement matrices $\matA_i$, we say that  a decoder $g$ achieves  error probability $\varepsilon$
if
\begin{align}
\Pr\mleft[g\big(\tp{(\langle\matA_1,\rmatX\rangle,\dots, \langle\matA_k\rmatX\rangle)}\big)\neq\rmatX\mright]\leq \varepsilon.\nonumber
\end{align}
%and call $N\in\naturals$ an $\varepsilon$-recovery guarantee if $k>N$ is sufficient for the existence of  a decoder achieving error probability $\varepsilon$. 
\end{dfn}

%Next, we give a formal definition of $\varepsilon$-support sets.% and introduce the Minkowski dimension of bounded sets.%, which are sets that contain the random matrix  $\rmatX\in\reals^{m\times n}$ with probability larger than or equal to $1-\varepsilon$.

\begin{dfn}\label{defsupportset}
For $\varepsilon\geq0$, we call a nonempty bounded set $\setS\subseteq\reals^{m\times n}$ an $\varepsilon$-support set of the random matrix $\rmatX\in\reals^{m\times n}$ if $\Pr[\rmatX\in\setS]\geq 1-\varepsilon$.  
\end{dfn}

%The operationally relevant ``description complexity" of the random matrix $\rmatX$ will turn out to be  
%the lower Minkowski dimension of its $\varepsilon$-support sets. 

%The s measured in terms of 
%In the following definition we introduce the Minkowski dimension of bounded sets in $\reals^{m\times n}$. 

\begin{dfn} (Minkowski dimension\footnote{This quantity is sometimes also referred to as box-counting dimension, which is the origin for the subscript B in the notation $\dim_\mathrm{B}(\cdot)$ used below.})\label{dfndim}
Let $\setS$ be a nonempty bounded set in $\reals^{m\times n}$.  The lower Minkowski dimension of $\setS$ is defined as 
\begin{align}
\underline{\dim}_\mathrm{B}(\setS)=\liminf_{\rho\to 0} \frac{\log N_\setS(\rho)}{\log \frac{1}{\rho}}\nonumber
\end{align}
and the upper Minkowski dimension  is 
\begin{align}
\overline{\dim}_\mathrm{B}(\setS)=\limsup_{\rho\to 0} \frac{\log N_\setS(\rho)}{\log \frac{1}{\rho}}\nonumber 
\end{align}
where $N_\setS$ denotes the covering number of $\setS$ given by
\begin{align}%\label{eq:coveringnumber}
N_\setS(\rho)=\min\Big\{k \in\naturals\mid \setS\subseteq\hspace*{-4truemm} \bigcup_{i\in\{1,\dots,k\}}\hspace*{-4truemm} \setB_{m\times n}(\matM_i,\rho),\ \matM_i\in \reals^{m\times n}\Big\}.\nonumber
\end{align}
If $\underline{\dim}_\mathrm{B}(\setS)=\overline{\dim}_\mathrm{B}(\setS)=: \dim_\mathrm{B}(\setS)$, we simply say that $\dim_\mathrm{B}(\setS)$ is the Minkowski dimension of $\setS$.
\end{dfn}

\section{Main results}
%The following result shows that the lower Minkowski dimension of an $\varepsilon$-support set yields an $\varepsilon$-recovery guarantee, i.e., an achievability condition for recovery with probability of error at most $\varepsilon$, for a.a. measurement matrices.
The following result formalizes the statement on the operationally relevant description complexity being given by the lower Minkowski dimensions of $\varepsilon$-support sets of $\rmatX$.

\begin{thm}\label{th1}
Let $\setS\subseteq\reals^{m\times n}$ be an $\varepsilon$-support set of $\rmatX\in\reals^{m\times n}$. Then, for Lebesgue a.a. measurement matrices $\matA_i$,  $i=1,\dots,k$, there exists  a decoder achieving error probability $\varepsilon$, provided that $k>\underline{\dim}_\mathrm{B}(\setS)$.  
% measurable decoder $g$ such that 
%\begin{align}
%\Pr\mleft[g\big(\tp{(\langle\matA_1,\rmatX\rangle,\dots, \langle\matA_k\rmatX\rangle)}\big)\neq\rmatX\mright]< \varepsilon.
%f\end{align}
\end{thm}
\begin{proof}
See Section \ref{proofths}.
\end{proof}

\begin{rem}\label{rem:intuition}
The central conceptual element in the proof of Theorem \ref{th1} is the following probabilistic null space property, first reported in \cite{stribo13} in the context of almost lossless analog signal separation.   
For a.a.  measurement matrices $\matA_i$, $i=1,\dots,k$, the dimension of the kernel of the mapping  
$\matX\mapsto \tp{(\langle\matA_1,\matX\rangle,\dots, \langle\matA_k,\matX\rangle)}$ is $mn-k$. 
%It turns out that the set $\setS$ intersects the kernel of this mapping at most trivially for a.a. $\matA_i$, $i=1,\dots,k$, if 
%$\underline{\dim}_\mathrm{B}(\setS) < k$. 
If the lower Minkowski dimension of a set $\setS$ is smaller than $k$, the set $\setS$ will intersect 
the kernel of this mapping  at most trivially.  
%Underlying this argument is the basic idea that two objects whose dimensions do not add up to at least the dimension of their ambient space, in general, do not intersect. 
%This probabilistic nullspace property allows us to prove uniqueness of the mapping $\matX\mapsto \tp{(\langle\matA_1,\matX\rangle,\dots, \langle\matA_k,\matX\rangle)}$ for $\matX\in$
What is remarkable here is that the notions of Euclidean dimension (for the kernel of the mapping) and of lower Minkowski dimension (for $\setS$) are compatible.   
\end{rem}

We next particularize Theorem \ref{th1} for low-rank matrices.  To this end, we first establish an upper bound on 
$\overline{\dim}_\mathrm{B}(\setS)$ for nonempty and bounded subsets of $\setM_r^{m\times n}$.

\begin{lem}\label{lemlr}
Let $\setS\subseteq\setM_r^{m\times n}$ be a nonempty bounded set. Then  
\begin{align}
\overline{\dim}_\mathrm{B}(\setS)\leq (m+n-r)r.\nonumber
\end{align}
\end{lem}
\begin{proof}
%\begin{proof}
%Let $\setN_i^{m\times n}$ denote the set of matrices $\matX\in\reals^{m\times n}$ with $\rank(\matX)= i$. 
We can decompose $\setM_r^{m\times n}$  according to 
\begin{align}
\setM_r^{m\times n}=\bigcup_{i=0}^r\setN_i^{m\times n}.\nonumber
\end{align}
By \cite[Ex. 5.30]{le00}, $\setN_i^{m\times n}$ is an embedded submanifold of $\reals^{m\times n}$ of dimension $(m+n-i)i$, $i=1,\dots,r$.
Let $\setI=\big\{i\in \{1,\dots,r\}\mid \setS\cap\setN_i^{m\times n}\neq\emptyset\big\}$.  
Then, for each $i\in\setI$, $\setS\cap\setN_i^{m\times n}$ is a nonempty bounded set and, therefore, 
$\overline{\dim}_\mathrm{B}(\setS\cap\setN_i^{m\times n})$ is well-defined. 
By \cite[Sec. 3.2, Properties (i) and (ii)]{fa90}, $\overline{\dim}_\mathrm{B}(\setS\cap\setN_i^{m\times n})\leq (m+n-i)i$, $i\in\setI$. Since the upper Minkowski dimension is finitely stable 
\cite[Sec. 3.2, Property (iii)]{fa90}, we get 
\begin{align}
\overline{\dim}_\mathrm{B}(\setS)
&=\max_{i\in\setI}\overline{\dim}_\mathrm{B}(\setS\cap\setN_i^{m\times n})\nonumber\\
%&\leq \max_{i\in\setI} (m+n-i)i\nonumber\\
%&\leq mn-(m-r)(n-r)\\
&\leq (m+n-r)r\nonumber
\end{align}
where  in the last step we used the monotonicity  of $f(s)=(m+n-s)s$ in the range  $s\in [0,(m+n)/2]$  together with $r\leq (m+n)/2$, which in turn follows from $r\leq\min(m,n)$. 
%\vspace*{-6truemm}
\end{proof}
%\end{proof}
%\section{Full measurement matrices}
%The following theorem links the Minkowski dimension of bounded sets with the number of measurements that are needed in the matrix recovery problem. Specifically, it  states that for Lebesgue a.a. measurement matrices $\matA_i$ there exists a decoder $g$ that achieves error probability $\varepsilon$ provided that the there exists a bounded set 
%$\setS\subseteq\reals^{m\times n}$ such that $\Pr[\matX\in\setS]\geq 1-\varepsilon$ and the number of measurements $k>\underline{\dim}_\mathrm{B}(\setS)$. 

%for every full-rank matrix B ? Rk?l, with k ?? l, and every rate R
%with R >RB(?), for a.a. A there is a measurable separator achieving error probability ?.
We can now put the pieces together to get the desired statement on low-rank matrices. 

%\vspace*{1truemm}
\begin{rem}\label{rem1}
Lemma \ref{lemlr} together with $\underline{\dim}_\mathrm{B}(\cdot)\leq \overline{\dim}_\mathrm{B}(\cdot)$, 
when used in Theorem \ref{th1}, implies that for $\rmatX\in\setM_r^{m\times n}$ and every $\varepsilon>0$,
there exists a decoder that achieves  error probability $\varepsilon$ 
for Lebesgue a.a. measurement matrices $\matA_i$, $i=1,\dots,k$, provided that $k>(m+n-r)r$.
%, i.e., for the  existence of  a decoder achieving error probability $\varepsilon$  for a.a. measurement matrices.
%we can recover matrices $\rmatX\in\setM_r^{m\times n}$ of  general distribution with probability of error less than $\varepsilon$ provided that  the number of measurements. 
 %and increasing the size of the support set $\setS$ in Theorem \ref{th1}.
%Since the definition of Minkowski dimension requires that the support sets are bounded, we have to intersect }
%For example, we can construct a sequence of support sets\footnote{Note that $\setM_r^{m\times n}$ is not a bounded set and the Minkowski dimension is only defined for bounded sets.} $\setS_L=\setM_r^{m\times n}\cap \setB_{m\times n}(\matzero,L)$, $L\in\naturals$. Then, $\Pr[\rmatX\notin\setS_L]$ can be made arbitrarily small by increasing $L$. 
\end{rem}
%We note that $\varepsilon$ in Corollary \ref{rem1} can be made arbitrarily small (but not equal to zero for a general distribution). 

While the sufficient condition $k>(m+n-r)r$ in Remark \ref{rem1} is intuitively appealing as $(m+n-r)r$ is the dimension of the manifold $\setN_r^{m\times n}$, it is actually the lower Minkowski dimensions of $\varepsilon$-support sets of $\rmatX\in\setM_r^{m\times n}$ that are of operational significance. Specifically, depending on the distribution of $\rmatX$, a  smaller (than $(m+n-r)r$) number of measurements may suffice for recovery of  $\rmatX$ with probability of error at most $\varepsilon$. %, for a.a. measurement matrices. 
The following example illuminates this phenomenon. 
%For rank-$r$ matrices with a specific distribution, $\varepsilon$-recovery guarantees  can be much smaller than $(m+n-r)r$. To this end, consider the following example.
%The lower Minkowski dimension of $\varepsilon$-support sets of rank-$r$ matrices can be much smaller than $(m+n-r)r$, which implies that  $  

%\begin{rem}\label{rem2}
%For a specific distribution of the matrices   $\rmatX\in\setM_r^{m\times n}$ and a.a. measurement matrices, the number of measurements $k$ that are needed to recover matrices $\rmatX\in\setM_r^{m\times n}$ with probability of error less than $\varepsilon$ can be much smaller than $k>(m+n-r)r$. Such a distribution is constructed in Example \ref{exa1}.
%\end{rem}

\begin{exa}\label{exa1}
Let $\rmatX =\tp{\rmatX_1}\rmatX_2\in \setM_r^{m\times n}$, where $\rmatX_1\in\reals^{r\times m}$ and  $\rmatX_2\in\reals^{r\times n}$ are independent. % and $r<\min(m/2,n/2)$. 
Suppose that 
$\rmatX_1$   has $l_1$ columns at positions drawn uniformly at random and containing i.i.d. Gaussian entries with all other columns equal to  zero and $\rmatX_2$  has $l_2$ columns at positions drawn uniformly at random and containing i.i.d. Gaussian entries with all other columns equal to  zero. 
Suppose further  that $r\leq l_1 < m/2$ and $r\leq l_2 \leq n/2-1/r$. The assumptions $l_i\geq r$, $i=1,2$,  guarantee that $\Pr[\rank(\rmatX)=r]=1$. 
%\begin{align}
%&\geq \Pr[\rank(\rmatX_1)=r,\rank(\rmatX_2)=r]\\
%&=\Pr[\rank(\rmatX_1)=r]\Pr[,\rank(\rmatX_2)=r]\\
%&1. 
%\end{align}
Next, we construct an $\varepsilon$-support set $\setT$ for $\rmatX$ with $\dim_\mathrm{B}(\setT)\leq (l_1+l_2)r$, which by Theorem \ref{th1} together with 
$(l_1+l_2)r<(m+n)r/2-1\leq (m+n-r)r-1$ proves that we can recover the rank-$r$ matrix 
$\rmatX$ with probability of error at most $\varepsilon$ from strictly less than $(m+n-r)r$ measurements.

%with a number of measurements that is strictly smaller than the dimension  of the manifold $\setN_r^{m\times n}$, given by $(m+n-r)r$.

%recovery with probability of error at most $\varepsilon$, for a.a. measurement matrices.

%we can get $\varepsilon$-recovery for 

Let 
$\setA^{r\times m}_{l}\subseteq \setM_r^{r\times m}$ be the  set of $r\times m$ matrices with no more than $l$ nonzero columns. 
%and $\setA^{r\times n}_{l_2}\subseteq \setM_r^{r\times n}$ the set of $r\times n$ matrices with no more than  $l_2$ nonzero columns. 
Choose $L\in\naturals$ sufficiently large for  (i) 
$\setS_1=\setA^{r\times m}_{l_1}\cap \setB_{r\times m}(\matzero,L)$  to be an  $\varepsilon/2$-support set of $\rmatX_1$ and 
(ii) $\setS_2=\setA^{r\times n}_{l_2}\cap \setB_{r\times n}(\matzero,L)$  to be an  $\varepsilon/2$-support set of $\rmatX_2$.
By  \cite[Sec. 3.2, Properties (i) and (iii)]{fa90}, 
 we have 
\begin{align}\label{eq:dimSi}
\dim_\mathrm{B}(\setS_i)=l_ir
\end{align}
which is simply the maximum number of nonzero entries of  $\matX_i\in\setS_i$, $i=1,2$. Set $\setT=\{\tp{\matX_1}\matX_2\mid \matX_i\in\setS_i, i=1,2\}$. 
%. According to Remark \ref{rem1}, $k=m+n-1$ measurements are sufficient to reconstr
Then,  
\begin{align}
\Pr[\rmatX\in\setT]
&=\Pr[\rmatX_1\in\setS_1, \rmatX_2\in\setS_2]\nonumber\\
&=\Pr[\rmatX_1\in\setS_1]\Pr[\rmatX_2\in\setS_2]\nonumber\\%\label{eq:indepexample}\\
&\geq 1-\varepsilon.\nonumber%\label{eq:supsetexample}
\end{align}
%where \eqref{eq:indepexample} follows from  $\rmatX_1$ and $\rmatX_2$ being statistically independent. 
%Since the sets $\setS_i$ are bounded, there exists a constant $L$ such that $\|\matX_i\|_2\leq L$ for all matrices $\matX_i\in\setS_i$, $i=1,2$. 
The triangle inequality implies  that for all $\matX_i,\bar\matX_i\in\setS_i$, $i=1,2$, we have 
\begin{align}
&\|\tp{\matX_1}\matX_2-\tp{\bar \matX_1}\bar\matX_2 \|_2\nonumber\\
&\leq\|\tp{\matX_1}\matX_2-\tp{\bar \matX_1}\matX_2 \|_2 +\|\tp{\bar \matX_1}\matX_2- \tp{\bar \matX_1}\bar\matX_2\|_2\nonumber\\
&\leq L(\|\matX_1-\bar\matX_1\|_2 +\|\matX_2- \bar\matX_2\|_2)\label{eq:trexa}
\end{align}
where we used $\setS_1\subseteq \setB_{r\times m}(\matzero,L)$ and $\setS_2\subseteq \setB_{r\times n}(\matzero,L)$. 
Let $N_{\setS_i}(\rho)$ be the covering number of $\setS_i$, $i=1,2$. 
We can cover $\setS_i$ by $N_{\setS_i}(\rho)$ balls of radius $\rho$ with centers $\bar\matX_{j_i}$, $j_i=1,\dots, N_{\setS_i}(\rho)$, $i=1,2$. 
Therefore, \eqref{eq:trexa} implies that  $\setT$ can be covered by $N_{\setS_1}(\rho)N_{\setS_2}(\rho)$ balls of radius $2L\rho$ centered at 
$\tp{\bar\matX_{j_1}}\bar\matX_{j_2}$, $j_i=1,\dots, N_{\setS_i}(\rho)$, $i=1,2$. This yields  
$N_\setT(2L\rho)\leq N_{\setS_1}(\rho)N_{\setS_2}(\rho)$ and we finally get 
%Let $N_\rho(\setS_i)$ be the covering number of $\setS_i$, $i=1,2$. Then \eqref{eq:trexa} 
% If we can cover $\setS_i$ by  balls of radius $\rho$, then \eqref{eq:trexa} implies that we can cover the set $\setT$ by  balls of radius $2L\rho$. Therefore, 
\begin{align}
\dim_\mathrm{B}(\setT)
&=\lim_{\rho\to 0} \frac{\log N_\setT(2L\rho)}{\log \frac{1}{2L\rho}}\nonumber\\
&\leq \lim_{\rho\to 0} \frac{\log\big(N_{\setS_1}(\rho)N_{\setS_2}(\rho)\big)}{\log \frac{1}{2L\rho}}\nonumber\\
&= \lim_{\rho\to 0} \frac{\log N_{\setS_1}(\rho)}{\log \frac{1}{2L\rho}}+\lim_{\rho\to 0} \frac{\log N_{\setS_2}(\rho)}{\log \frac{1}{2L\rho}}\nonumber\\
&=l_1r+l_2r\nonumber%\label{eq:dimexa}
\end{align}
where we used \eqref{eq:dimSi} in the last step. 
%Combining \eqref{eq:supsetexample} and \eqref{eq:dimexa}, we obtain that $\setT$ is an $\varepsilon$-support set for $\rmatX$ with $\dim_\mathrm{B}(\setT)\leq  r(l_1+l_2)$. By Theorem \ref{th1}, $\dim_\mathrm{B}(\setT)$ is an $\varepsilon$-recovery guaranty. Therefore, also $r(l_1+l_2)$ is an 
%$\varepsilon$-recovery guaranty. Now $r\leq \min(m,n)$ implies  $r\leq (n+m)/2$ and, in turn,  $r(m+n)/2 \leq r(m+n-r)$. Since $r(l_1+l_2)<r(m+n)/2$ by assumption, we obtain that the $\varepsilon$-recovery guaranty $r(l_1+l_2)$ is  strictly smaller than $r(m+n-r)$.
%which implies that 
%Therefore, the number of measurements $k$ can be smaller than $r(m+n-r)+1$.
% Therefore, by Theorem \ref{th1}, we can recover $\rmatX$ with a.a. measurement matrices and probability of error less than $\varepsilon$, provided that the number of measurements fulfills  $k>r(l_1+l_2)$, which is strictly smaller than $r(m+n-r)$.
\end{exa}

\begin{rem}
The derivation of the recovery thresholds in  \cite{elnepl12} is also based on a null space property similar to the one discussed in Remark \ref{rem:intuition}. The relevant dimension in \cite{elnepl12} is the  dimension $(m+n-r)r$ of the manifold $\setN_r^{m\times n}$.  %as opposed to the lower Minkowski dimension of $\varepsilon$-support sets. 
Example \ref{exa1} above, however, shows that $k< (m+n-r)r$ measurements can suffice  for recovery of rank-$r$ matrices, thereby corroborating the operational significance of the lower  Minkowski dimensions of  $\varepsilon$-support sets of $\rmatX$. 
%for a specific distribution of matrices $\rmatX\in \setM_r^{m\times n}$, the 
%lower Minkowski dimension of $\varepsilon$-support sets can be much smaller than the  dimension of $\setM_r^{m\times n}$. 
\end{rem}

%\vspace*{-1truemm}

\section{Rank-one measurement matrices}

Rank-one measurement matrices, i.e., matrices  $\matA_i=\veca_i\tp{\vecb_i}$ with  $\veca_i\in\reals^m$ and $\vecb_i\in\reals^n$, $i=1,\dots,k$,  are attractive as they 
%exhibit much less richness than general measurement matrices and, therefore, 
require less storage space than general measurement matrices and can also be applied faster.  
%Specifically, for $m\leq n$, each measurement $\tr(\tp{\matA_i}\matX)=\tp{\veca_i}\matX\vecb_i$ requires  the computation of only $m(n+1)$ products  as opposed to $m^2n$ products for general measurement matrices. 
%In addition, $(m+n)$ real numbers are sufficient to store $\veca_i$ and $\vecb_i$, in contrast to $mn$ real numbers required for a general measurement matrix  $\matA_i$. 
% This leads to a significant reduction in terms of computational cost and storage space. 
 %Since  $\setM_1^{m\times n}$  is a set of measure zero with respect to  Lebesgue measure on $\reals^{m\times n}$, it might be entirely in the the set of s exceptional set of measure zero where 
% Theorem \ref{th1} does not apply. 
%We now derive  $\varepsilon$-recovery guarantees for rank-one measurement matrices. 
%it is possible to achieve the same recovery guarantees as in Theorem \ref{th1} by using rank one projections instead of full measurement matrices. We now derive $\varepsilon$-recovery guarantees for rank one measurement matrices.  
Interestingly, %even though rank-one measurement matrices , 
Theorem \ref{th1} continues to hold for rank-one measurement matrices although they exhibit much less richness than general measurement matrices. %$\matA_i=\veca_i\tp{\vecb_i}$. 
%Interestingly, Theorem \ref{th1} continues to hold for rank-one measurement matrices, which have much less richness than general measurement matrices.  
%$\matA_i=\veca_i\tp{\vecb_i}\in\setM_1^{m\times n}$ and Lebesgue a.a. $\veca_i\in\reals^m$ and $\vecb_i\in\reals^n$. 
The technical challenges in establishing this result are quite different from those encountered in the case of general measurement matrices. In particular, we will need a stronger concentration of measure inequality (cf. Lemma \ref{lem:com1}).  
%(\eqref{eq:com inequality}) 
%than in the case of general measurement matrices. 

\begin{thm}\label{th2}
Let $\setS\subseteq\reals^{m\times n}$ be an $\varepsilon$-support set of $\rmatX\in\reals^{m\times n}$. Then, for Lebesgue a.a.  $\veca_i\in\reals^m$ and $\vecb_i\in\reals^n$ and corresponding measurement matrices $\matA_i=\veca_i\tp{\vecb_i}$, $i=1,\dots,k$, there exists a decoder achieving error probability $\varepsilon$, provided that $k>\underline{\dim}_\mathrm{B}(\setS)$.
\end{thm}
\begin{proof}
See Section \ref{proofths}.
\end{proof}
%\begin{rem}
%It is remarkable that it is possible to derive the same recovery threshold with  
%\end{rem}
%\begin{rem}
%Noting that rank-one measurement matrices can be represented using $k(m+n)$ real numbers, as as opposed to $mnk$ real numbers for full measurement matrices, it is surprising to see that the same $\varepsilon$-recovery guarantees can be derived as in Theorem \ref{th1}. 
%\end{rem}

%\begin{rem}
%The central tenets behind the proof of Theorem \ref{th2} is again a null space property. However, in the case of rank-one measurement matrices the 
% proof is much more involved  and requires the derivation of a new concentration of measure inequality, which is formulated in Lemma  \ref{lem:com}.
%\end{rem}

\begin{rem}
Example \ref{exa1} can be shown to carry over to rank-one measurement matrices $\matA_i$. %, which implies that $k\leq (m+n-r)r$ measurements can suffice  for $\varepsilon$-recovery of rank-$r$ matrices using rank-one $\matA_i$. 
%also apply to the setting of Theorem \ref{th2}. 
%Therefore, $\varepsilon$-recovery guarantees for a specific distribution can also be much smaller than $(m+n-r)r$ for rank-one measurement matrices.
\end{rem}
\begin{rem}
Theorem \ref{th2}, %holds for Lebesgue a.a.  vectors $\veca_i$ and $\vecb_i$, $i=1,\dots,k$, and, 
when used in combination with Lemma \ref{lemlr}, implies that for $\rmatX\in\setM_r^{m\times n}$ and every $\varepsilon>0$,
there exists a decoder achieving error probability $\varepsilon$ 
for Lebesgue a.a.  $\veca_i\in\reals^m$ and $\vecb_i\in\reals^n$, provided that $k>(m+n-r)r$.  
  %and does not make specific assumptions on their distribution.  
In contrast, the threshold $k\geq C(m+n)r$ in \cite{cazh15} for rank-one measurements requires $\rveca_i$ and $\rvecb_i$ to be independent random vectors containing i.i.d. Gaussian or sub-Gaussian entries. In addition, the constant $C$ in \cite{cazh15} remains unspecified. %revaluating our $\varepsilon$-recovery guarantee $k>(m+n-r)$ %in Theorem \ref{th2} stronger. 
%requires the entries of these vectors to be  i.i.d.  Gaussian or sub-Gaussian distributed.  Moreover, the recovery threshold $k\geq Cr(m+n)$ derived in \cite{cazh15}  depends on an unknown constant $C$. 
\end{rem}

\section{Proofs of Theorems  \ref{th1} and \ref{th2}}\label{proofths}
For both proofs, we first construct  
%In the case of Theorem \ref{th2}, we can set $\matA_i=\vecb_i\tp{\veca_i}$, $1,\dots,k$, which allows us to write $\tp{\veca_i}\matX\vecb_i
%=\langle\matA_i,\matX\rangle,\quad i=1,\dots,k$. 
 a measurable map
%\footnote{We need measurability of $g$ to guarantee that the image of a random variable under  $g$ is again a random variable.} 
$g:\reals^{k}\to \reals^{m\times n}$ such that 
\begin{align}
&\Pr\mleft[[g\big(\tp{(\langle\matA_1,\rmatX\rangle,\dots, \langle\matA_k,\rmatX\rangle)}\big)\not\neq\rmatX\mright]\nonumber \\
&\leq \Pr\mleft[\exists\matZ\in\setS_\rmatX\!\setminus\!\{\matzero\} \big| \tp{(\langle\matA_1,\matZ\rangle,\dots, \langle\matA_k,\matZ\rangle)}=\veczero,\rmatX\in\setS \mright]+\varepsilon \label{eq:stepone}
\end{align}
with $\setS_\matX=\{\matW-\matX\mid \matW\in\setS\}$ for $\matX\in\setS$. 
%This step is  identical for Theorems \ref{th1} and \ref{th2}.
The proofs are then concluded by showing that 
\begin{align}\label{eq:havetoshow}
\Pr\mleft[\exists\matZ\in\setS_{\rmatX}\!\setminus\!\{\matzero\} \big| \tp{(\langle\matA_1,\matZ\rangle,\dots, \langle\matA_k,\matZ\rangle)}=\veczero,\rmatX\in\setS \mright]=0
\end{align}
for Lebesgue a.a.  matrices $\matA_i\in\reals^{m\times n}$, $i=1,\dots,k$, in the case of  Theorem \ref{th1} and for Lebesgue a.a.  vectors $\veca_i\in\reals^{m}$ and $\vecb_i\in\reals^{n}$ with $\matA_i=\veca_i\tp{\vecb_i}$, $i=1,\dots,k$, in the case of  Theorem \ref{th2}. 
%This step is significantly more difficult in the case of rank-one measurement matrices.
%which concludes the proof of Theorem \ref{th2}. Note that the probability in \eqref{eq:havetoshow} is the probability of nonuniqueness of matrices $\rmatX\in\setS$ given certain linear measurements of $\rmatX$.   

\emph{Proof of \eqref{eq:stepone}}: Let  $\setS\subseteq\reals^{m\times n}$ be an $\varepsilon$-support set of $\rmatX$ with 
$\underline{\dim}_\mathrm{B}(\setS)< k$. %and $\Pr[\matX\in\setS]> 1-\varepsilon$. 
We define a measurable map $g$ as follows:
\begin{align}
&g(\vecy)=\nonumber\\
&\begin{cases}
\matZ,&  \text{if}\ \big\{\matW\in\setS \mid \tp{(\langle\matA_1,\matW\rangle,\dots, \langle\matA_k,\matW\rangle)}=\vecy\big\} = \{\matZ\}\\
\matE, & \text{else}
\end{cases}\nonumber
\end{align}
where $\matE$ is an arbitrary, but fixed, matrix in $\reals^{m\times n}\setminus\setS$ 
(used to declare a decoding error). 
%representing an error event. %The matrix $\matE$ is chosen if it is not possible to identify a unique matrix $\matW\in\setS$ that is mapped to $\vecy$. 
Then, we have 
\begin{align}
&\Pr\mleft[g\big(\tp{(\langle\matA_1,\rmatX\rangle,\dots, \langle\matA_k,\rmatX\rangle)}\big)\not\neq\rmatX\mright]\nonumber\\
%&=
%\Pr\mleft[g\big(\tp{(\langle\matA_1,\rmatX\rangle,\dots, \langle\matA_k,\rmatX\rangle)}\big)\not\neq\rmatX,\rmatX\in\setS\mright]\nonumber\\
%&\phantom{=}+\Pr\mleft[g\big(\tp{(\langle\matA_1,\rmatX\rangle,\dots, \langle\matA_k,\rmatX\rangle)}\big)\not\neq\rmatX,\rmatX\notin\setS\mright]\\
&\leq 
\Pr\mleft[g\big(\tp{(\langle\matA_1,\rmatX\rangle,\dots, \langle\matA_k,\rmatX\rangle)}\big)\not\neq\rmatX,\rmatX\in\setS\mright]+\Pr\mleft[\rmatX\notin\setS\mright]\nonumber\\
&\leq
\Pr\mleft[g\big(\tp{(\langle\matA_1,\rmatX\rangle,\dots, \langle\matA_k,\rmatX\rangle)}\big)\not\neq\rmatX,\rmatX\in\setS\mright]
+\varepsilon\label{eq:errorbound1a}\\
&=\Pr\mleft[g\big(\tp{(\langle\matA_1,\rmatX\rangle,\dots, \langle\matA_k,\rmatX\rangle)}\big)=\matE,\rmatX\in\setS\mright]
+\varepsilon\label{eq:errorbound1b}\\
&=\Pr\mleft[\exists\matZ\in\setS_\rmatX\!\setminus\!\{\matzero\} \big| \tp{(\langle\matA_1,\matZ\rangle,\dots, \langle\matA_k,\matZ\rangle)}=\veczero,\rmatX\in\setS \mright]+\varepsilon\nonumber%\label{eq:errorbound1}
\end{align}
where \eqref{eq:errorbound1a} is a consequence of  $\matS$ being an $\varepsilon$-support set and in \eqref{eq:errorbound1b} we used that the decoder declares an error if and only if $|\{\matW\in\setS \mid \tp{(\langle\matA_1,\matW\rangle,\dots, \langle\matA_k,\matW\rangle)}=\vecy\}|>1$ for 
$\vecy=\tp{(\langle\matA_1,\matX\rangle,\dots, \langle\matA_k,\matX\rangle)}$ with $\matX\in\setS$.
%for $\rmatX\in\setS$ the decoder produces an error if and only if it can not find a unique element in $\setS$.  

\emph{Finishing the proof of Theorem 1}: 
Let $s>0$ and 
suppose that $\rmatA_1,\dots,\rmatA_k$, $i=1,\dots,k$, are %such that  the vectors $\operatorname{vec}(\rmatA_i)$ are 
independent and uniformly distributed on $\setB_{m\times n}(\veczero,s)$.  Then, we have  
\begin{align}
&\hspace*{-8truemm}\int\limits_{\hspace*{9truemm}\big(\setB_{m\times n}(\matzero,s)\big)^{k}}\hspace*{-11truemm} \Pr\mleft[\exists \matZ\in\setS_\rmatX\!\setminus\!\{\matzero\}\mid \tp{(\langle\matA_1,\matZ\rangle,\dots, \langle\matA_k,\matZ\rangle)} =\matzero,\rmatX\in\setS\mright]\nonumber\\
&\hspace*{9truemm}\operatorname{d}\!\mu_{\rmatA_1}\times\dots\times\operatorname{d}\!\mu_{\rmatA_k}\nonumber\\
&=\int\limits_{\setS} \Pr\mleft[\exists \matZ\in\setS_\matX\!\setminus\!\{\matzero\}\mid \tp{(\langle\rmatA_1,\matZ\rangle,\dots, \langle\rmatA_k,\matZ\rangle)} =\matzero\mright]\operatorname{d}\!\mu_\rmatX\label{eq:fubini}\\
&=0\label{eq:propzero}
\end{align}
where \eqref{eq:fubini} is a consequence of Fubini's theorem for  nonnegative measurable  functions and \eqref {eq:propzero} follows from Lemma \ref{lem:probzero1} below. 
%noting that the function involved is nonnegative and Lebesgue measurable. 
With $\reals^{m\times n}=\bigcup_{s\in\naturals}\setB_{m\times n}(\veczero,s)$ and since $s$ is arbitrary,   \eqref{eq:havetoshow} holds for Lebesgue a.a. measurement matrices $\matA_i$, which concludes the proof of Theorem \ref{th1}. 

\emph{Finishing the proof of Theorem 2}:  Let $s>0$ and suppose that $\rmatA=[\rveca_1,\dots,\rveca_k]\in \reals^{m\times k}$ and $\rmatB=[\rvecb_1,\dots,\rvecb_k]\in \reals^{n\times k}$ are independent random matrices with columns $\rveca_i$ independent and uniformly distributed on $\setB_m(\veczero,s)$ and columns $\rvecb_i$ independent and uniformly distributed on $\setB_n(\veczero,s)$.
Then, we have  
\begin{align}
&\hspace*{-11truemm}\int\limits_{\hspace*{12truemm}\big(\setB_{m}(\matzero,s)\times\setB_{n}(\matzero,s)\big)^k}\hspace*{-16truemm} \Pr\mleft[\exists \matZ\in\setS_\rmatX\!\setminus\!\{\matzero\}\mid \tp{(\tp{\veca_1}\matZ\vecb_1,\dots,\tp{\veca_k}\matZ\vecb_k)} =\matzero,\rmatX\in\setS\mright]\nonumber\\
&\hspace*{9truemm}
\operatorname{d}\!\mu_{\rveca_1}\times\operatorname{d}\!\mu_{\rvecb_1}\times\dots\times\operatorname{d}\!\mu_{\rveca_k}\times\operatorname{d}\!\mu_{\rvecb_k}\nonumber\\
&=\int\limits_{\setS} \Pr\mleft[\exists \matZ\in\setS_\matX\!\setminus\!\{\matzero\}\mid \tp{(\tp{\rveca_1}\matZ\rvecb_1,\dots,\tp{\rveca_k}\matZ\rvecb_k)} =\matzero\mright]\operatorname{d}\!\mu_\rmatX\label{eq:fubini1}\\
&=0\label{eq:propzero1}
\end{align}
where \eqref{eq:fubini1} is a consequence of Fubini's theorem for nonnegative measurable functions and \eqref {eq:propzero1} follows from Lemma \ref{lem:probzero2} below.  
%noting that the function involved is nonnegative and Lebesgue measurable. 
Again, with $\reals^{l}=\bigcup_{s\in\naturals}\setB_{l}(\veczero,s)$ and since $s$ is arbitrary,   
\eqref{eq:havetoshow} holds for Lebesgue a.a. vectors  $\veca_i\in\reals^m$ and $\vecb_i\in\reals^n$, thereby finishing the proof of Theorem \ref{th2}.\qed
%%%%%%%%%%%%%%%%%%%%%%%%%%%%%%%%%%%%%%%%%%%%%%%%%%%%%%%%%%%%%%%%%%%%%%%%%%%%%%%%%%%%%%%%%%%%%%%%%%%%%%%%%%%%%%%%%%%%%%%%%%%%%%%%%%%%%%%%%%%%%%%%%%
\begin{lem}\label{lem:probzero1}
Let $s>0$ and $\rmatA_1,\dots,\rmatA_k$, $i=1,\dots,k$, be 
%such that the vectors $\operatorname{vec}(\rmatA_i)$, $i=1,\dots,k$, are 
independent and uniformly distributed on $\setB_{m\times n}(\veczero,s)$.  
Suppose that $\setU\subseteq \reals^{m\times n}$ is a nonempty bounded set with $\underline{\dim}_\mathrm{B}(\setU)<k$. Then, we have 
\begin{align}\label{eq:propfin1}
\Pr\mleft[\exists \matX\in\setU\!\setminus\!\{\matzero\}\mid \tp{(\langle\rmatA_1,\matX\rangle,\dots, \langle\rmatA_k,\matX\rangle)} =\matzero\mright]=0.
\end{align} 
\end{lem}
\begin{proof}
Follows from rewriting the trace inner products $\langle\rmatA_i,\matX\rangle$, $i=1,\dots,k$ as inner products between vectors in $\reals^{mn}$ and subsequent  application of \cite[Prop. 1]{stribo13}. 
%Let $\setV=\{\operatorname{vec}(\matX)\mid \matX\in\setU\}$. Since the mapping $(\reals^{m\times n},\|\cdot\|_2)\to (\reals^{mn},\|\cdot\|_2)$, $\matX\mapsto \operatorname{vec}(\matX)$ is an isometry, $\setV$ is a bounded set with $\underline{\dim}_\mathrm{B}(\setV)=\underline{\dim}_\mathrm{B}(\setU)$. Setting 
%\begin{align}
%\rmatA=
%\begin{pmatrix}
%\tp{\operatorname{vec}(\rmatA_1)}\\
%\vdots\\
%\tp{\operatorname{vec}(\rmatA_k)}
%\end{pmatrix}\nonumber
%\end{align}
%we can rewrite \eqref{eq:propfin1} as 
%\begin{align}%\label{eq:propfin2}
%\Pr\mleft[\exists \vecx\in\setV\!\setminus\!\{\veczero\}\mid \rmatA\vecx =\matzero\mright]=0\nonumber
%\end{align} 
%and apply \cite[Prop. 1]{stribo13} to finish the proof.
\end{proof}

\begin{lem}\label{lem:probzero2}
Let $s>0$ and take $\rmatA=[\rveca_1,\dots,\rveca_k]\in \reals^{m\times k}$ and $\rmatB=[\rvecb_1,\dots,\rvecb_k]\in \reals^{n\times k}$ to be  independent random matrices with columns $\rveca_i$, $i=1,\dots,k$, independent and uniformly distributed on $\setB_m(\veczero,s)$ and columns $\rvecb_i$, $i=1,\dots,k$, independent and uniformly distributed on $\setB_n(\veczero,s)$. Suppose that $\setU\subseteq \reals^{m\times n}$ is a nonempty bounded set with $\underline{\dim}_\mathrm{B}(\setU)<k$. Then, we have 
\begin{align}%\label{eq:propfin}
\!P&:=\Pr\mleft[\exists \matX\in\setU\!\setminus\!\{\matzero\}\big|\tp{(\tp{\rveca_1}\matX\rvecb_1,\dots,\tp{\rveca_k}\matX\rvecb_k)}=\matzero\mright]=0.\nonumber
\end{align} 
\end{lem}

%\vspace*{-2truemm}

\begin{proof} 
%\subsection{Proof of Lemma \ref{lem:probzero2}}\label{pr:probzero2}
Let $R=\max_{\matX\in\setU}\rank(\matX)$
%\begin{align}
%R=\max_{\matX\in\setU}\rank(\matX)\nonumber
%\end{align}
and set 
\begin{align}
\setU_{L,r}=\Big\{\matX\in\setU\mid \Delta(\matX) >\frac{1}{L}, \sigma_1(\matX)<L, \rank(\matX)=r\Big\}\nonumber
\end{align}
for $L\in\naturals$ and  $r=1,\dots, R$.
By the union bound, we have  
\begin{align}\label{eq:seriesprob}
P\leq
\sum_{L\in\naturals}\sum_{r=1}^R
P_{L,r} 
\end{align}
where 
\begin{align}
P_{L,r} =\Pr\mleft[\exists \matX\in\setU_{L,r}\mid \tp{(\tp{\rveca_1}\matX\rvecb_1,\dots,\tp{\rveca_k}\matX\rvecb_k)}=\matzero\mright].\nonumber%\label{PLr}
\end{align}
We now prove by contradiction that $P_{L,r}=0$ for all $L\in\naturals$ and all $r\in\{1,\dots,R\}$. %By definition,  $P_{L,r}\geq 0$ for all $L\in\naturals$ and $r=1,\dots, R$. 
Suppose  
 that there exists an  $L\in\naturals$ and an $r\in\{1,\dots,R\}$ such that $P_{L,r}>0$ (by definition,  $P_{L,r}\geq 0$). 
 For this pair $\{L,r\}$, we would then have
 %Then, we have  
\begin{align}
\liminf_{\rho\to 0}\frac{\log P_{L,r}}{\log\frac{1}{\rho}}=0.\label{eq:liminf0}
\end{align}
%and $\setU_{L,r}\neq\emptyset$. 
For $\rho >0$,  let $N_{\setU_{L,r}}(\rho)$ be the  covering number of the set $\setU_{L,r}$ and  
denote corresponding covering balls centered at $\matM_i(\rho)\in\reals^{m\times n}$ 
 as
 $\setB_{m\times n}(\matM_i(\rho),\rho)$,  $i=1,\dots,N_{\setU_{L,r}}(\rho)$. %  balls that cover $\setU_{L,r}$ with  %having nonempty intersection with $\setU_{L,r}$. 
We now fix $N_{\setU_{L,r}}(\rho)$ matrices 
\begin{align}\label{eq:Xi}
\matX_i(\rho)\in \setB_{m\times n}(\matM_i(\rho),\rho)\cap \setU_{L,r},\quad i=1,\dots,N_{\setU_{L,r}}(\rho). 
\end{align}
Since 
\begin{align}
\setB_{m\times n}(\matM_i(\rho),\rho)\subseteq \setB_{m\times n}(\matX_i(\rho),2\rho),\quad i=1,\dots,N_{\setU_{L,r}}(\rho)\nonumber
\end{align}
by the triangle inequality,  we get
\begin{align}
P_{L,r}
&\leq \hspace*{-3truemm}\sum_{i=1}^{N_{\setU_{L,r}}(\rho)}\hspace*{-2truemm}
\Pr\!\big[\exists \matX\in \setB_{m\times n}(\matM_i(\rho),\rho) \!\mid\nonumber\\
&\phantom{\leq \hspace*{-3truemm}{N_{\setU_{L,r}}(\rho)}\hspace*{-2truemm}
\Pr[}
\tp{(\tp{\rveca_1}\matX\rvecb_1,\dots,\tp{\rveca_k}\matX\rvecb_k)}=\matzero\big]\nonumber\\
&\leq \hspace*{-3truemm}\sum_{i=1}^{N_{\setU_{L,r}}(\rho)}\hspace*{-2truemm}
\Pr\!\big[\exists \matX\in \setB_{m\times n}(\matX_i(\rho),2\rho)  \!\mid \nonumber\\
&\phantom{\leq \hspace*{-3truemm}{N_{\setU_{L,r}}(\rho)}\hspace*{-2truemm}
\Pr[}
\tp{(\tp{\rveca_1}\matX\rvecb_1,\dots,\tp{\rveca_k}\matX\rvecb_k)}=\matzero\big]\nonumber\\
&\leq \hspace*{-3truemm}\sum_{i=1}^{N_{\setU_{L,r}}(\rho)}\hspace*{-2truemm}
\Pr\!\big[\exists \matX\in \setB_{m\times n}(\matX_i(\rho),2\rho) \! \mid \nonumber\nonumber\\
&\phantom{\leq \hspace*{-3truemm}{N_{\setU_{L,r}}(\rho)}\hspace*{-2truemm}
\Pr[}
\|\tp{(\tp{\rveca_1}\matX\rvecb_1,\dots,\tp{\rveca_k}\matX\rvecb_k)}\|_2\leq\rho\big],\quad\rho>0\label{eq:bound1}.
\end{align}
Now, for $\veca_i\in \setB_m(\veczero,s)$, $\vecb_i\in \setB_n(\veczero,s)$, and $\matX\in \setB_{m\times n}(\matX_i(\rho),2\rho)$, we have %, with  
\begin{align}
&\|\tp{(\tp{\veca_1}\matX_i(\rho)\vecb_1,\dots,\tp{\veca_k}\matX_i(\rho)\vecb_k)}\|_2\nonumber\\ 
&\leq \|\tp{(\tp{\veca_1}(\matX-\matX_i(\rho))\vecb_1,\dots,\tp{\veca_k}(\matX-\matX_i(\rho))\vecb_k)}\|_2\nonumber\\
&\phantom{\leq}+\|\tp{(\tp{\veca_1}\matX\vecb_1,\dots,\tp{\veca_k}\matX\vecb_k)}\|_2\nonumber\\
%&=\sqrt{\sum_{j=1}^k (\tp{\veca_j}(\matX-\matX_i(\rho))\vecb_j)^2}\nonumber\\
%&\phantom{\leq}+\|\tp{(\tp{\veca_1}\matX\vecb_1,\dots,\tp{\veca_k}\matX\vecb_k)} \|_2\nonumber\\
&\leq\sqrt{\sum_{j=1}^k \|\veca_j\|_2^2\|\matX-\matX_i(\rho)\|_2^2\|\vecb_j\|_2^2}\nonumber\\
&\phantom{\leq}+\|\tp{(\tp{\veca_1}\matX\vecb_1,\dots,\tp{\veca_k}\matX\vecb_k)} \|_2\nonumber\\
&\leq 2s^2\sqrt{k}\rho +\|\tp{(\tp{\veca_1}\matX\vecb_1,\dots,\tp{\veca_k}\matX\vecb_k)} \|_2,\quad\rho>0.\label{eq:boundLS}
\end{align}
%where in \eqref{eq:boundLS} we used that  
%We  therefore established that 
%\begin{align}\label{eq:triangprop}
%&\|\tp{(\tp{\rveca_1}\matX\rvecb_1,\dots,\tp{\rveca_k}\matX\rvecb_k)}\|_2\nonumber\\
%&\geq \|\tp{(\tp{\rveca_1}\matX_i(\rho)\rvecb_1,\dots,\tp{\rveca_k}\matX_i(\rho)\rvecb_k)}\|_2- 2s^2\sqrt{k}\rho,\quad\rho>0 . 
%\end{align}
Inserting \eqref{eq:boundLS} %\eqref{eq:triangprop} 
into \eqref{eq:bound1} allows us to further upper-bound  
$P_{L,r}$ according to 
\begin{align}
P_{L,r}
&\leq \sum_{i=1}^{N_{\setU_{L,r}}(\rho)}
\Pr\!\big[\|\tp{(\tp{\rveca_1}\matX_i(\rho)\rvecb_1,\dots,\tp{\rveca_k}\matX_i(\rho)\rvecb_k)}\|_2\nonumber\\
&\phantom{\leq {N_{\setU_L}(\rho)}\Pr[}
\leq\rho(1+2s^2\sqrt{k})\big]\nonumber\\
&\leq 2^{\frac{k(m+n)}{2}-kr} N_{\setU_{L,r}}(\rho)\rho^{k}g(L,r,k,s,\rho)^{k},\quad\rho>0\\\label{eq:bound2}
\end{align}
where
\begin{align}
&g(L,r,k,s,\rho)=\\
&L(1+2s^2\sqrt{k})\\
&\times
\begin{cases}
\frac{2}{s^2}  +\frac{2}{s^2}\log\max\Big(\frac{s^2L}{\rho},1\Big), & \text{if}\ r=1\\
\frac{V(r,1)(\rho(1+2s^2\sqrt{k}))^{r-1}}{s^{2r}} +\frac{A(r-1,1)L^{r-1}}{s^2(r-1)},& \text{if}\ r>1.
\end{cases}
\end{align}
Here, we applied the concentration of measure inequality in Lemma \ref{lem:com1} below with $\delta=\rho(1+2s^2\sqrt{k})$ and used the fact that $1/\Delta(\matX_i(\rho)) < L$, $\sigma_1(\matX_i(\rho))<L$, and $\rank(X_i(\rho))=r$ (recall that by \eqref{eq:Xi} all matrices $\matX_i(\rho)$ are in the set $\setU_{L,r}$). 
With the upper bound on $P_{L,r}$ in \eqref{eq:bound2} we now get 
\begin{align}
&\liminf_{\rho\to 0}\frac{P_{L,r}}{\log\frac{1}{\rho}}\\
&\leq\liminf_{\rho\to 0}\frac{\log(N_{\setU_{L,r}}(\rho))+k\log\rho+k\log g(L,r,k,s,\rho)}{\log\frac{1}{\rho}}\\
&=\liminf_{\rho\to 0}\frac{\log(N_{\setU_{L,r}}(\rho))}{\log\frac{1}{\rho}}-k
+\lim_{\rho\to 0}\frac{k\log g(L,r,k,s,\rho)}{\log\frac{1}{\rho}}
\\
&=\liminf_{\rho\to 0}\frac{\log(N_{\setU_{L,r}}(\rho))}{\log\frac{1}{\rho}}-k
\\
&\leq 
\liminf_{\rho\to 0}\frac{\log(N_{\setU}(\rho))}{\log\frac{1}{\rho}}-k\label{eq:subU}\\
&=\underline{\dim}_\mathrm{B}(\setU)-k\\
&<0\label{eq:contr}
\end{align}
where \eqref{eq:subU} follows from $\setU_{L,r}\subseteq \setU$ and in the last step we used that  $\underline{\dim}_\mathrm{B}(\setU)<k$, by assumption. 
Since \eqref{eq:contr} contradicts \eqref{eq:liminf0}, $P_{L,r}=0$ for all $L\in\naturals$ and all $r\in\{1,\dots,R\}$.  
By \eqref{eq:seriesprob}, this establishes that $P=0$. \end{proof}

\begin{lem}\label{lem:com1} 
Let $\rmatA=[\rveca_1,\dots, \rveca_k]$ and $\rmatB=[\rvecb_1,\dots, \rvecb_k]$ be independent random matrices, with columns $\rveca_i$, $i=1,\dots,k$, independent and uniformly distributed on $\setB_m(\veczero,s)$ and  columns $\rvecb_i$, $i=1,\dots,k$, independent and uniformly distributed on $\setB_n(\veczero,s)$. 
Suppose that $\matX\in\reals^{m\times n}$ with  $r=\rank(\matX)>0$.
Then, we have  
\begin{align}%\label{eq:com inequality}
\opP\mleft[\big\|\tp{(\tp{\rveca_1}\matX\rvecb_1, \dots,\tp{\rveca_k}\matX\rvecb_k)}\big\|_2\leq \delta\mright]\leq\delta^{k} 2^{\frac{k(m+n)}{2}-kr}f(\matX,s,\delta)^{k} 
\end{align}
with $f(\matX,s,\delta)$ defined in \eqref{eq:f}.
\end{lem}
\begin{proof} 
We have 
\begin{align}
&\opP\mleft[\big\|\tp{(\tp{\rveca_1}\matX\rvecb_1, \dots, \tp{\rveca_k}\matX\rvecb_k)}\big\|_2\leq \delta\mright]\\
&=\opP\mleft[\sum_{i=1}^k(\tp{\rveca}_i\matX\rvecb_i)^2\leq \delta^2\mright]\\
&\leq \opP\mleft[|\tp{\rveca}_i\matX\rvecb_i|\leq \delta,\ \text{for all}\ i=1,\dots,k\mright]\\
&=
\opP\mleft[|\tp{\rveca}\matX\rvecb|\leq \delta\mright]^{k} \label{eq:unioncom1}\\
&\leq\delta^{k}  2^{\frac{k(m+n)}{2}-kr} f(\matX,s,\delta)^{k}\label{eq:applycom1}
\end{align}
where in \eqref{eq:unioncom1}  $\rveca$ and  $\rvecb$ are independent with $\rveca$ uniformly distributed on $\setB_m(\veczero,s)$ and $\rvecb$ uniformly distributed on $\setB_n(\veczero,s)$ and, therefore, we can apply Lemma \ref{lem:com} below to obtain  \eqref{eq:applycom1}.
\end{proof}

\begin{lem}\label{lem:com}
Let $\rveca$ and $\rvecb$ be independent random vectors, with  $\rveca$ uniformly distributed on $\setB_m(\veczero,s)$ and 
$\rvecb$ uniformly distributed on $\setB_n(\veczero,s)$. 
Suppose that $\matX\in\reals^{m\times n}$ with $r=\rank(\matX)>0$. 
Then, we have  
\begin{align}
\opP[|\tp{\rveca}\matX\rvecb|\leq \delta]
&\leq 
\delta D_{r,m,n}f(\matX,s,\delta)
\end{align}
where\footnote{We use the convention that $V(0,s)=1$.}
\begin{align}
D_{r,m,n}
&=\frac{2V(n-r,1)V(m-r,1)V(r-1,1)}{V(m,1)V(n,1)}\\
%&=\frac{\Gamma\big(\frac{m}{2}+1\big)\Gamma\big(\frac{n}{2}+1\big)}{(\sqrt{\pi}s)^{r+1}\Gamma\big(\frac{n-r}{2}+1\big)\Gamma\big(\frac{m-r}{2}+1\big)\Gamma\big(\frac{r-1}{2}+1\big)}\\
&\leq 2^{\frac{m+n}{2}-r}
\label{eq:D}
\end{align}
and
\begin{align}
&f(\matX,s,\delta)\\
&=\frac{1}{\Delta(\matX)}
\begin{cases}
\frac{2}{s^2}+\frac{2}{s^2}\log\max\Big(\frac{s^2\sigma_1(\matX)}{\delta},1\Big)& \text{if}\ r=1\\
\frac{\delta^{r-1}V(r,1)}{s^{2r}}+\frac{A(r-1,1)\sigma_1(\matX)^{r-1}}{s^2(r-1)}& \text{if}\ r>1.
\end{cases}\\
\label{eq:f}
\end{align} 
\end{lem}
\begin{proof}
Using Fubini's theorem for nonnegative measurable functions and noting that $1/V(m,s)$ and $1/V(n,s)$ is a probability density function for $\rveca$ and $\rvecb$, respectively, we can rewrite
\begin{align}
\opP[|\tp{\rveca}\matX\rvecb|\leq \delta]
&=\frac{1}{V(m,s)V(n,s)}\int_{\setB_m(\veczero,s)}h(\veca)\mathrm d\lebmeasure^{m}(\veca)\label{eq:comstep1}
\end{align}
with 
\begin{align}
h(\veca)
&=\int_{\setB_n(\veczero,s)}\ind{\{\vecb\in\reals^n : |\tp{\veca}\matX\vecb|\leq \delta\}}(\vecb)\mathrm d\lebmeasure^{n}(\vecb).\label{eq:h1}
\end{align}
Let $\matX=\matU\matSigma\matV$ be a singular value decomposition of $\matX$, where 
$\matU\in\reals^{m\times m}$ and 
$\matV\in\reals^{n\times n}$ are  orthogonal and 
\begin{align}
\matSigma=
\begin{pmatrix}
\matD&\matzero\\
\matzero&\matzero
\end{pmatrix}
\in\reals^{m\times n}
\end{align}
with $\matD=\diag(\sigma_1(\matX)\ldots\sigma_r(\matX))$. 
Using the fact that Lebesgue measure on  $\setB_m(\veczero,s)$ and $\setB_n(\veczero,s)$ is invariant under rotations, we can rewrite
\begin{align}
\opP[|\tp{\rveca}\matX\rvecb|\leq \delta]
&=\frac{1}{V(m,s)V(n,s)}\int_{\setB_m(\veczero,s)}h(\matU\veca)\mathrm d\lebmeasure^{m}(\veca)\\\label{eq:comstep2}
\end{align}
and 
\begin{align}
h(\matU\veca)
&=\int_{\setB_n(\veczero,s)}\ind{\{\vecb\in\reals^n : |\tp{\veca}\matSigma\vecb|\leq \delta\}}(\vecb)\mathrm d\lebmeasure^{n}(\vecb).\label{eq:h2}
\end{align} 
Decomposing $\veca=\tp{(\tp{\veca_1}\ \tp{\veca_2})}$ and $\vecb=\tp{(\tp{\vecb_1}\ \tp{\vecb_2})}$, with 
$\veca_1, \vecb_1\in\reals^r$, we can 
upper-bound $h(\matU\veca)$ by 
\begin{align}
h(\matU\veca)
&\leq V(n-r,s)\int_{\setB_r(\veczero,s)}\ind{\{\vecb_1\in\reals^r :|\tp{\veca_1}\matD\vecb_1|\leq \delta\}}(\vecb_1)\mathrm d\lebmeasure^{r}(\vecb_1)\\
&=\frac{V(n-r,s)}{\Delta(\matX)}\int_{\setB_r(\veczero,s\sigma_1(\matX))}\ind{\{\vecc\in\reals^r :|\tp{\veca_1}\vecc|\leq \delta\}}(\vecc)\mathrm d\lebmeasure^{r}(\vecc) \label{eq:h3}
\end{align}
where in the last step we changed variables to $\vecc=\matD\vecb_1$ and used that $\|\vecc\|_2\leq  \sigma_1(\matX)\|\vecb_1\|_2$.
Using  that Lebesgue measure on $\setB_r(\veczero,s\sigma_1(\matX))$ is invariant under rotations and setting $\vece_1=\tp{(1\ 0\ldots0)}\in\reals^r$, we can further upper-bound $h(\matU\veca)$ by 
\begin{align}
&h(\matU\veca)\\
&=
\frac{V(n-r,s)}{\Delta(\matX)}\int_{\setB_r(\veczero,s\sigma_1(\matX))}\ind{\big\{\vecc\in\reals^r :|\tp{\vece_1}\vecc|\leq \frac{\delta}{\|\veca_1\|_2}\big\}}(\vecc)\mathrm d\lebmeasure^{r}(\vecc)\\
&\leq
2V(n-r,s)V(r-1,s)\frac{\sigma_1(\matX)^{(r-1)}}{\Delta(\matX)} \min\Big(s\sigma_1(\matX),\frac{\delta}{\|\veca_1\|_2}\Big).
\label{eq:h4}
\end{align}
Plugging \eqref{eq:h4} into \eqref{eq:comstep2}, we find that 
\begin{align}
&\opP[|\tp{\rveca}\matX\rvecb|\leq \delta]\\
&\leq D_{r,m,n}\frac{\sigma_1(\matX)^{(r-1)}}{\Delta(\matX)s^{r+1}}\int_{\setB_r(\veczero,s)}\min\Big(s\sigma_1(\matX),\frac{\delta}{\|\veca_1\|_2}\Big)\mathrm d\lebmeasure^{r}(\veca_1).\label{eq:comstep3}
\end{align}
%where 
%\begin{align}\label{eq:CXs}
%C_{\matX,s}=\frac{2V(n-r,s)V(m-r,s)V(r-1,s\sigma_1(\matX))}{V(m,s)V(n,s)\Delta(\matX)}.  
%\end{align}
It remains to upper-bound the integral in \eqref{eq:comstep3}. We can split 
\begin{align}
\int_{\setB_r(\veczero,s)}\min\Big(s\sigma_1(\matX),\frac{\delta}{\|\veca_1\|_2}\Big)\mathrm d\lebmeasure^{r}(\veca_1)
&=I_1+I_2\label{eq:comstep4}
\end{align}
where 
\begin{align}
I_1
&=s\sigma_1(\matX)\int_{\setB_r(\veczero,s)\cap\setB_r\big(\veczero,\frac{\delta}{s\sigma_1(\matX)}\big)}\mathrm d\lebmeasure^{r}(\veca_1)\\
&\leq s\sigma_1(\matX)V\Big(r,\frac{\delta}{s\sigma_1(\matX)}\Big)\\
&=\frac{\delta^rV(r,1)}{(s\sigma_1(\matX))^{r-1}}
\end{align}
and
\begin{align}
I_2&=\delta\int_{\setB_r(\veczero,s)\cap\Big(\reals^s\setminus\setB_r\big(\veczero,\frac{\delta}{s\sigma_1(\matX)}\big)\Big)} \frac{\mathrm d\lebmeasure^{r}(\veca_1)}{\|\veca_1\|_2}.
\end{align}
For $r=1$ we get 
\begin{align}
I_2=2\delta \log\max\Big(\frac{s^2\sigma_1(\matX)}{\delta},1\Big).
\end{align}
If $r>1$  we can change variables to polar coordinates and find that 
\begin{align}
I_2
&=\delta A(r-1,1)\int_{\frac{\delta}{s\sigma_1(\matX)}}^{s}
\vecyc^{r-2} d\lebmeasure^{1}(\vecyc)\\
&\leq
\delta A(r-1,1)\int_{0}^{s}
\vecyc^{r-2} d\lebmeasure^{1}(\vecyc)\\
&=\delta\frac{A(r-1,1)s^{(r-1)}}{r-1}.
\end{align}
Therefore, we have 
\begin{align}
I_1+I_2
\leq 
\delta\times 
\begin{cases}
2+2\log\max\Big(\frac{s^2\sigma_1(\matX)}{\delta},1\Big)& \text{if}\ r=1\\
\frac{\delta^{r-1}V(r,1)}{(s\sigma_1(\matX))^{r-1}}+\frac{A(r-1,1)s^{r-1}}{r-1}& \text{if}\ r>1.
\end{cases}\label{eq:I12}
\end{align}
Combining \eqref{eq:comstep3},   \eqref{eq:comstep4}, and \eqref{eq:I12} we finally end up with 
\begin{align}
\opP[|\tp{\rveca}\matX\rvecb|\leq \delta]
&\leq 
\delta D_{r,m,n}f(\matX,s,\delta).
\end{align}
The upper bound on $D_{r,m,n}$ follows from $2^{k/2}<V(k,1)<2^k$ for $k\in\naturals$.
\end{proof}

%\vspace*{-1truemm}

%\linespread{0.95}

%\begin{proof} See Section  \ref{pr:com1} .\end{proof}  

%\proof See Section  \ref{pr:probzero2}.\end{proof}  
%%%%%%%%%%%%%%%%%%%%%%%%%%%%%%%%%%%%%%%%%%%%%%%%%%%%%%%%%%%%%%%%%%%%%%%%%%%%%%%%%%%%%%%%%%%%%%%%%%%%%%%%%%%%%%%%%%%%%%%%%%%%%%%%%%%%%%%%%%%%%%%%%%
%\addbibresource{IEEEabrv.bib}
\bibliographystyle{IEEEtran}
%\bibliography{IEEEabrv,references,isit15}
\bibliography{../../../../../jabref/references}
%\printbibliography
%\newpage
%%%%%%%%%%%%%%%%%%%%%%%%%%%%%%%%%%%%%%%%%%%%%%%%%%%%%%%%%%%%%%%%%%%%%%%%%%%%%%%%%%%%%%%%%%%%%%%%%%%%%%%%%%%%%%%%%%%%%%%%%%%%%%%%%%%%%%%%%%%%%%%%%%
\end{document}